\documentclass[journal]{IEEEtran}
\usepackage{amsmath,amsfonts,amssymb,bm}
\usepackage{amsthm}
\usepackage{xcolor}
\usepackage{array}
\usepackage{algorithm}
\usepackage{algpseudocode}
\usepackage{url}
\usepackage{graphicx,flushend}
\usepackage[none]{hyphenat}
\newcommand{\argmin}{\mathop{\rm argmin}}

\newcommand{\nmode}[2]{\big[\mathbf{\ten{#1}}\big]_{\left(#2\right)}}
\newcommand{\sumx}[2]{\sum\limits_{#1}^{#2}}
\newcommand{\bb}[1]{\mathbb{#1}}
\usepackage{bm}
\newcommand{\ten}[1]{\bm{\mathcal{#1}}}

\newcommand{\unvec}[1]{\text{unvec}\big( #1 \big)}

\newcommand{\algwrap}[1]{\parbox[t]{\dimexpr\linewidth-\algorithmicindent\relax}{#1}}
\newcommand{\revblue}[1]{{\color{black}#1}}

\newtheorem{proposition}{Proposition}
\newtheorem{corollary}{Corollary}

\begin{document}

\title{TenSIM: Tensor-Based Channel Estimation for MIMO Systems with Stacked Intelligent Metasurfaces}
\author{André L. F. de Almeida and George C. Alexandropoulos}

\maketitle

\begin{abstract}
Stacked intelligent metasurfaces (SIMs) are emerging as a promising architecture for the sixth generation (6G) and beyond of wireless systems, enabling richer electromagnetic-wave manipulation than conventional single-layer metasurfaces. However, realizing these gains requires accurate and scalable channel estimation under the strong inter-layer coupling and multilinear parameter interactions introduced by the stacked programmable metasurface layers. This paper proposes TenSIM, a tensor-based channel-estimation framework for SIM-assisted multiple-input multiple-output (MIMO) systems. By exploiting a structured SIM training protocol, TenSIM derives two parity-dependent observation models: a PARAllel FACtor (PARAFAC) model for odd-layer SIMs and a Tucker model for even-layer SIMs. These formulations decouple the transmitter-SIM and SIM-receiver channel factors while explicitly accounting for inter-layer wave coupling. Based on the resulting tensor models, we develop alternating least squares estimators, establish rank-based identifiability conditions using the associated design matrices, and characterize practical sufficient conditions for full-column-rank training designs, including those involving scaling ambiguities. The proposed framework is validated through extensive numerical experiments and reveals the main operating trade-offs. Our study shows that both TenSIM-PARAFAC and TenSIM-Tucker improve with the signal-to-noise ratio and the training diversity, outperforming unstructured least-squares baselines by exploiting the tensor structure of the SIM cascade. We also show that TenSIM-PARAFAC offers better scalability, lower computational complexity, and stronger robustness to inter-layer spacing, while TenSIM-Tucker can provide more accurate channel reconstruction when sufficient training and strong layer coupling are available. Finally, it is shown that the proposed TenSIM framework remains effective under imperfect or blind SIM training when additional pilot diversity is available. Overall, TenSIM provides a unified and physically interpretable approach for channel estimation in SIM-assisted MIMO systems with explicit identifiability, complexity, and performance trade-offs.
\end{abstract}

\begin{IEEEkeywords}
Stacked intelligent metasurfaces, TenSIM, channel estimation, tensor decomposition, PARAFAC, Tucker model, alternating least squares, identifiability.
\end{IEEEkeywords}

\maketitle

\section{Introduction}
Programmable metasurfaces are emerging as a key enabler for future wireless networks, introducing additional degrees of freedom in the propagation environment when incorporated~\cite {RIS_Book}. Instead of treating multipath as an uncontrollable impairment, metasurface technologies enable software-defined wavefront manipulation through programmable subwavelength elements, allowing phase, amplitude, polarization, and even dispersion engineering with low hardware complexity and power consumption \cite{ElMossallamy2020}. This paradigm is particularly attractive for dense, blockage-prone, and high-frequency deployments, where conventional fully active relaying can be prohibitively expensive in terms of radio-frequency chains and energy budgets~\cite{RIS_use_cases}. In this context, reconfigurable intelligent surfaces (RISs) have received extensive attention due to their potential to enhance coverage, energy efficiency, and link reliability~\cite{Basar_2024}. A substantial body of work has established the communication-theoretic foundations and signal processing methods for RIS-assisted systems, including joint active/passive beamforming~\cite{Huang2019, WuZhang2019, Guo2019}, optimization under practical constraints~\cite{Rahal2022, Bjornson2022,Gavriilidis2025}, localization~\cite{RIS_localization,RISloc_AP_free,RIS_smartcities}, integrating sensing and communications~\cite{RIS_ISAC}, as well as channel estimation~\cite{RIS_survey}.

Tensor modeling has also played an important role in signal processing for wireless systems, owing to its ability to compactly represent and exploit the inherent multilinear structure of communication signals and channels~\cite{almeida2007parafac,confac,favier2014tensor,favier2014overview}. More recently, tensor decompositions have been successfully applied to channel estimation and semi-blind receiver (Rx) design in RIS-assisted wireless communications systems~\cite{gil2021,gil2022, wEI_ris_1,YUAN_ris_1,wEI_ris_2,ardah2021trice,nwalozie2025doubleRIS,Andre_Almeida_2025}.

Most existing studies on programmable metasurfaces focus on single-layer reflective architectures, whose wave-control capability is ultimately limited by the aperture size and one-shot surface interaction. To overcome this limitation, stacked intelligent metasurfaces (SIM)~\cite{SIM_paper}, also referred to as multilayer or volumetric programmable surfaces, have been introduced to achieve richer wave transformations through cascaded interlayer coupling.
By cascading multiple programmable layers of metamaterials separated by sub-wavelength distances, SIMs can emulate deeper linear transformations in the electromagnetic domain~\cite{StylianopoulosOTA2026}, yielding enhanced beamforming flexibility, stronger wave focusing, and improved multiplexing capability compared with conventional single-layer RIS structures. For recent survey and tutorial papers on SIM principles, applications, and open challenges, we refer the reader to~\cite{HLiu2025SIMSurvey,EShi2026SIMSurvey,MLiu2026InWave,Semantic_Alignment,MINN_power_control,SIM_CST}.

\revblue{Recent research on programmable metasurfaces has also expanded along several complementary directions. These include discrete XL-RIS codebook design and near-field multiuser interference management~\cite{XL_RIS_Codebook_2025}, low-complexity meta-fiber-assisted SIM architectures~\cite{MetaFiberSIM2025}, practical varactor-based RIS hardware with continuous phase control~\cite{VaractorRIS2025}, joint precoder and reflector design for RIS-assisted orbital-angular-momentum communications~\cite{RIS_OAM_2025}, visually steered RIS-aided mobile communications~\cite{VisuallySteeredRIS2025}, and joint power allocation and discrete phase-shift optimization for SIM-aided integrated sensing and communications (ISAC)~\cite{SIM_ISAC_Power2025}. In addition, implicit CSI based approaches have been investigated for RIS precoding, including learning-based capacity-driven designs~\cite{CapacityNetRIS2024}.
}

The increased over-the-air computational power of SIM~\cite{SIM_MINN2026} comes with a major signal-processing challenge: channel estimation becomes a highly coupled, multi-factor inverse problem, since the end-to-end response includes the channels from the transmitter (Tx) to the SIM and from the SIM to the Rx, inter-layer coupling matrices, and layer-dependent phase profiles. Direct estimation with naive pilot designs typically leads to high-dimensional bilinear/multilinear ambiguities and poor identifiability. Therefore, channel estimation constitutes a key problem in wireless systems, including SIM. Without an accurate estimate of the cascaded factors, the additional degrees of freedom offered by stacked layers of response-tunable metamaterials cannot be reliably translated into gains in beamforming, multiplexing, and coverage. More specifically, compared with conventional single-layer RISs, SIM introduces stronger inter-layer coupling, higher-dimensional parameter spaces, and more severe bilinear/multilinear ambiguities between propagation channels and phase profiles, which make pilot design, identifiability, and estimator complexity fundamentally more challenging.

An important structural property of programmable metasurface cascades is that the multilinear model induced by the training process depends fundamentally on whether the number of SIM layers is odd or even. This distinction directly determines the appropriate tensor decomposition for channel estimation. When the SIM cascade contains an odd number of layers, a single programmable SIM layer at the center of the cascade, whose phase shifts vary across the training blocks, is enough for channel estimation, while the remaining layers remain fixed. In this case, the inter-layer propagation matrices on the left and right sides of the central layer can be grouped into two effective matrices.
In contrast, when the SIM cascade has an even number of layers, the training process involves two programmable layers positioned near the center of the cascade. These layers interact through the inter-layer propagation matrices associated with the middle inter-layer channel. 
As will be shown later in this paper, in the first case, the received pilot signals follow a parallel factor (PARAFAC) tensor model whose factor matrices correspond to the effective transmit and receive channels and to the programmable phase configuration matrix. In the second case, the resulting model takes the form of a Tucker decomposition with a known core structure determined by the central inter-layer propagation matrix. Such structural differences have a direct implication on the channel estimation for programmable metasurface cascades. The odd-layer SIM architectures naturally lead to PARAFAC-based estimation methods, whereas even-layer SIM architectures are more appropriately handled using Tucker decompositions with known core tensors. 

In this paper, the developed estimation methods exploit precisely this structural distinction, which has not been explicitly exploited in previous SIM channel estimation studies, motivating the proposed tensor-based estimation framework. \revblue{This parity distinction should be interpreted in conjunction with the proposed central-layer training protocol. 
If several noncentral layers were reconfigured simultaneously, the received data could lead to higher-order or coupled tensor models. Under the low-reconfiguration protocol adopted here, however, the physical position of the trainable middle layer(s) determines the algebraic structure. More specifically, one middle layer produces diagonal modulation between two effective sub-cascades, whereas two middle layers leave a central propagation operator between the trained phase profiles and therefore induce a Tucker core.}

To date, only a limited number of works have explicitly addressed the aforementioned problem. The work in \cite{CEWCL2024} provides one of the first SIM-specific channel estimation formulations and develops a structured estimator for cascaded SIM channels. The method of~\cite{SparseAPWCS2024} targeted millimeter-wave SIM settings and exploited channel sparsity to reduce training overhead and improve recovery efficiency. The deep learning approach in \cite{DLCE2025} learns the nonlinear SIM input-output mapping from data, thereby enhancing estimation performance under complex channel conditions. In \cite{NestedCE2025}, a nested-tensor method leveraged multiway algebraic structure to improve factor separation and channel reconstruction in SIM-assisted links. However, the resulting solution is computationally demanding and estimates all inter-layer SIM transfer functions, which is often unnecessary in practice, since inter-layer propagation is typically well-behaved and can be accurately modeled using known geometry/electromagnetic calibration. Despite these advances, existing methods still face limitations, including dependence on restrictive channel assumptions (e.g., strong sparsity), sensitivity to training design, significant data/computation demands in learning-based schemes, as well as residual identifiability and scalability issues in large-scale SIM deployments.

The current state of SIM channel estimation motivates structured training strategies that inject training diversity where it is most informative while keeping SIM hardware reconfiguration overhead manageable. 
In particular, in this paper, we are interested in SIM training protocols that preserve practical controllability while inducing a favorable algebraic structure for estimation. Building on this design, we formulate TenSIM as a unified tensor-based channel-estimation framework for SIM-assisted MIMO communication systems. The key step is a Kronecker-structured decomposition of the SIM cascade that separates the two directional sub-cascades and yields parity-dependent tensor models for the SIM transfer tensor and, consequently, for the received pilot tensor. In the odd-layer case, TenSIM specializes in a structured PARAFAC formulation, while in the even-layer case, it specializes in a Tucker formulation with a known core induced by the central inter-layer propagation. This representation is critical because it enables decoupled recovery of Tx-SIM and SIM-Rx channel factors and provides identifiability guarantees under mild rank/diversity conditions. Therefore, TenSIM jointly addresses physical interpretability, training efficiency, and estimation tractability in SIM-assisted Multiple-Input Multiple-Output (MIMO) communication systems. In summary, the main contributions of this paper are listed as follows:
\begin{itemize}
\item We introduce a novel training strategy for the SIM metasurface architecture that injects diversity at the structure's most informative layers, while keeping hardware reconfiguration overhead manageable.
\item We introduce TenSIM, a unified tensor-based channel-estimation framework comprising two parity-dependent models: a PARAFAC formulation for odd-layer SIMs and a Tucker formulation for even-layer SIMs, both grounded in the physics and structural components of the SIM cascade.
\item We develop practical TenSIM estimators based on alternating least squares (ALS), tailored to each tensor model, enabling structured and decoupled recovery of the Tx-SIM and SIM-Rx channel factors.
\item We establish identifiability conditions in terms of design matrix ranks and provide constructive sufficient conditions linked to full-rank training and the SIM cascade structure.
\item We quantify computational complexity and highlight the trade-off between the proposed lighter TenSIM-PARAFAC estimator and the richer, but more expensive, TenSIM-Tucker estimator.
\item We numerically evaluate the performance of both estimators, highlighting their trade-offs, as well as comparing them with representative baseline methods.
\end{itemize}

The remainder of this paper is organized as follows. Section~II reviews the tensor preliminaries and notation used throughout the paper. Section~III presents the system model and the proposed SIM training protocol. Section~IV develops the TenSIM tensor formulations for odd- and even-layer SIMs and discusses the associated training matrix designs. Section~V formulates the proposed TenSIM channel estimation algorithms, and Section~VI studies identifiability. Section~VII reports the numerical results, while Section~VIII concludes the paper.

\section{Tensor preliminaries}
In this section, we provide the notation and main operators used throughout this paper and an overview of the Tucker decomposition, which will be used to develop the TenSIM channel estimation algorithms.

\subsection{Notation and Key Matrix Properties} \label{Sec:notation}
Scalars are represented as non-bold lower-case letters (e.g., $a$), column vectors as lower-case boldface letters (e.g., $\mathbf{a}$), matrices as upper-case boldface letters (e.g., $\mathbf{A}$), sets as calligraphic lower-case letters (e.g., $\mathcal{A}$), and tensors as calligraphic upper-case boldface letters (e.g., $\ten{A}$). The superscripts $\{\cdot\}^{\text{T}}$, $\{\cdot\}^{\text{*}}$, $\{\cdot\}^{\text{H}}$, and $\{\cdot\}^{+ }$ stand for transpose, conjugate, conjugate transpose, and Moore-Penrose pseudo-inverse operations, respectively. \textcolor{black}{An identity matrix of dimension $K$ is denoted as $\mathbf{I}_{K}$.} The operator $\lVert\cdot\rVert_{F}$ denotes the Frobenius norm of a matrix or tensor. The operator $\text{D}(\mathbf{a})$ constructs a diagonal matrix holding vector $\mathbf{a}$ on the main diagonal, while given a matrix $\mathbf{A} \in \bb{C}^{I \times R}$, $\text{D}_{i}(\mathbf{A})$ defines a diagonal matrix constructed from the $i$-th row of $\mathbf{A}$, for $i \in \{1,\ldots,I\}$.
The operator $\text{vec}\left(\mathbf{A}\right)$ converts $\mathbf{A} \in \mathbb{C}^{I \times R}$ to a column vector $\mathbf{a} \in \mathbb{C}^{IR \times 1}$ by stacking its columns on top of each other, while $\unvec{\mathbf{a}}_{I \times R}$ returns to the matrix $\mathbf{A} \in \bb{C}^{I \times R}$. The symbols $\otimes$ and  $\diamond$ denote the Kronecker product and the Khatri-Rao product (also known as the column-wise Kronecker product), respectively. 
The following identity will be used repeatedly:
\begin{equation}\label{eq:vecABC}
\mathrm{vec}(\mathbf A \mathbf X \mathbf B)
=
(\mathbf B^T \otimes \mathbf A)\,\mathrm{vec}(\mathbf X).
\end{equation}
If $\mathbf X=\mathrm{D}(\mathbf x)$, then it follows:
\begin{equation}
\mathrm{vec}(\mathbf A\,\mathrm{D}(\mathbf x)\,\mathbf B)
=
(\mathbf B^T \diamond \mathbf A)\mathbf x.
\end{equation}

\subsection{Tensors, Slices, and Unfoldings}
Consider a set of matrices $\mathbf{Y}_{k} \in \bb{C}^{I \times J}$, \textcolor{black}{$\forall k = 1,\ldots, K$.} Concatenating all $K$ matrices, we form the third-order tensor $\ten{Y} \triangleq \mathbf{Y}_1 \sqcup_3 \mathbf{Y}_2 \sqcup_3 \ldots \sqcup_3 \mathbf{Y}_{K} \in \bb{C}^{I \times J \times K}$, where $\sqcup_3 $ indicates a concatenation along the third dimension. We can interpret $\mathbf{Y}_{k}$ as the $k$-th frontal slice of $\ten{Y}$, defined as the matrix $\ten{Y}_{..k} = \mathbf{Y}_{k} \in \bb{C}^{I \times J}$. This matrix is constructed by varying the first and second dimensions while keeping the third-dimension index $k$ fixed. The tensor $\ten{Y}$ can be \textit{matricized} by letting one dimension vary along the rows and the remaining two dimensions along the columns. From $\ten{Y}$, we can form three different matrices, referred to as the \textit{$n$-mode unfolding}, $n=1,2,3$, which can be respectively obtained as a function of the frontal slices as
\begin{align}
\label{eq:nmode_1}\nmode{Y}{1} &= [\ten{Y}_{..1},\ldots,\ten{Y}_{..K}] \in \bb{C}^{I \times JK}, \\  
\label{eq:nmode_2}\nmode{Y}{2} &= [\ten{Y}_{..1}^{\text{T}},\ldots,\ten{Y}_{..K}^{\text{T}}] \in \bb{C}^{J \times IK},\\ 
\label{eq:nmode_3}\nmode{Y}{3} &= [\text{vec}(\ten{Y}_{..1}),\ldots,\text{vec}(\ten{Y}_{..K})]^{\text{T}} \in \bb{C}^{K \times IJ}.
\end{align}

For convenience, we can also refer to the unfolding operation as $\nmode{Y}{n}=\textrm{unfold}(\ten Y,n)$, $n=1,2,3$.
The $n$-mode product, denoted as ``$\times_n$", defines the multiplication between a tensor $\ten{Y}$ and a matrix $\mathbf{A}$, leading to a tensor $\ten{Z}$ with compatible dimensions, i.e., $\ten{Z} \triangleq \ten{Y} \times_n \mathbf{A}$. It can be computed by pre-multiplying the $n$-mode unfolding of $\ten{Y}$ by the matrix $\mathbf{A}$, i.e.,  $[\ten{Z}]_{(n)} = \mathbf{A}[\ten{Y}]_{(n)}$. 
For example, the mode-1 product between $\ten{Y}\in\mathbb{C}^{I\times J\times K}$ and $\mathbf{A}\in\mathbb{C}^{L\times I}$ yields $\ten{Z} =\ten{Y} \times_1 \mathbf{A}\in\mathbb{C}^{L\times J\times K}$. It can be computed by $[\ten{Z}]_{(1)} = \mathbf{A}[\ten{Y}]_{(1)}\in\mathbb{C}^{L\times JK}$. We refer the interested reader to \cite{Sidiropoulos2017} for an overview.

\subsection{Tucker Decomposition}\label{Sec:Tucker_decom}
The Tucker decomposition \cite{Kolda2009} defines the concept of multilinear transformation.
For a third-order tensor $\ten{Y} \in \bb{C}^{I \times J \times K}$, this decomposition expresses the tensor as multiple sums of rank-one tensor components, which can be defined using a tensor notation as follows:
\begin{align}
  \label{eq:tenZ_tucker_generic} \ten{Y} &= \ten{G} \times_1 \mathbf{A} \times_2 \mathbf{B} \times_3 \mathbf{C}. 
\end{align}
Using the outer product notation, where $\circ$ denotes the vector outer product, the Tucker decomposition can be equivalently written as
\begin{align}
\ten{Y}
&=
\sumx{r_1=1}{R_1} \sumx{r_2=1}{R_2} \sumx{r_3=1}{R_3}
 g_{r_1,r_2,r_3}
 \mathbf{a}_{r_1} \circ \mathbf{b}_{r_2} \circ \mathbf{c}_{r_3},
\end{align}
where $\mathbf{a}_{r_1}$, $\mathbf{b}_{r_2}$, and $\mathbf{c}_{r_3}$ are the $r_1$-th, $r_2$-th, and $r_3$-th columns of $\mathbf{A}$, $\mathbf{B}$, and $\mathbf{C}$, respectively, and $g_{r_1,r_2,r_3}\triangleq[\ten{G}]_{r_1,r_2,r_3}$ denotes the typical element of $\ten{G}$.

The $3$-mode (frontal) slices $\ten{Y}_{..k} \in \mathbb{C}^{I \times J}$, $\forall k=1, \ldots, K$, can be expressed as follows:
\begin{align}
   \label{eq:tucker_2_slice} \ten{Y}_{..k} &= \mathbf{A}\big(\ten{G} \times_3 \mathbf{c}_k^{\text{T}}\big)\mathbf{B}^{\text{T}} \in \bb{C}^{I  \times J},
\end{align}
where $\mathbf{c}_k^{\text{T}}$ denotes the $k$-th row of $\mathbf{C}$.
By properly stacking these frontal slices according to equations (\ref{eq:nmode_1})--(\ref{eq:nmode_3}), the  matrix unfoldings of the Tucker decomposition can be factorized as
\begin{align}
     \label{eq:tucker_2_u1} \nmode{Y}{1} &= \mathbf{A}\nmode{G}{1}\big(\mathbf{C} \otimes \mathbf{B} \big)^{\text{T}} \in \bb{C}^{I \times JK}, \\
    \label{eq:tucker_2_u2} \nmode{Y}{2} &= \mathbf{B}\nmode{G}{2}\big(\mathbf{C} \otimes \mathbf{A} \big)^{\text{T}} \in \bb{C}^{J \times IK}, \\
    \label{eq:tucker_2_u3} \nmode{Y}{3} &= \mathbf{C}\nmode{G}{3}\big(\mathbf{B} \otimes \mathbf{A} \big)^{\text{T}} \in \bb{C}^{K \times IJ}.
\end{align}

\subsection{PARAFAC Decomposition}\label{Sec:parafac_decom}
The well-known PARAFAC decomposition, also known as the canonical polyadic decomposition (CPD)  \cite{Harshman}, \cite{Sidiropoulos2017}, is a special case of the Tucker decomposition, in which $R_1=R_2=R_3=R$, while the core tensor reduces to an identity tensor. In $n$-mode product form, for a third-order tensor $\ten{Y}\in\bb{C}^{I\times J\times K}$, the following is deduced:
\begin{equation}
\ten{Y} = \ten{I}_{3,R} \times_1 \mathbf{A} \times_2 \mathbf{B} \times_3 \mathbf{C},
\end{equation}
where $\mathbf{A}=[\mathbf{a}_1,\ldots,\mathbf{a}_R]\in\bb{C}^{I\times R}$, $\mathbf{B}=[\mathbf{b}_1,\ldots,\mathbf{b}_R]\in\bb{C}^{J\times R}$, and $\mathbf{C}=[\mathbf{c}_1,\ldots,\mathbf{c}_R]\in\bb{C}^{K\times R}$ are the factor matrices, and $R$ is the PARAFAC rank (number of rank-one components). Equivalently, it holds:
\begin{equation}
\ten{Y}=\sum_{r=1}^{R}\mathbf{a}_r\circ\mathbf{b}_r\circ\mathbf{c}_r,
\end{equation}
which shows that the PARAFAC decomposition expresses a higher-order tensor as a sum of triplets, each of which is a rank-one tensor (an outer product of three vectors). Adopting a similar reasoning as done before, the matrix unfoldings of a three-way PARAFAC tensor admit the following factorizations:
\begin{align}
\label{eq:parafac_u1}
\nmode{Y}{1} &= \mathbf{A}\left(\mathbf{C}\diamond\mathbf{B}\right)^{\text{T}} \in \bb{C}^{I\times JK},\\
\label{eq:parafac_u2}
\nmode{Y}{2} &= \mathbf{B}\left(\mathbf{C}\diamond\mathbf{A}\right)^{\text{T}} \in \bb{C}^{J\times IK},\\
\label{eq:parafac_u3}
\nmode{Y}{3} &= \mathbf{C}\left(\mathbf{B}\diamond\mathbf{A}\right)^{\text{T}} \in \bb{C}^{K\times IJ}.
\end{align}

The Tucker decomposition is generally non-unique due to rotational freedom between the factor matrices and the core tensor. Specifically, each factor matrix can be post-multiplied by a nonsingular transformation, provided that the inverse transformation is absorbed into the corresponding mode of the core tensor, yielding the same reconstructed tensor \cite{Kolda2009}. When the core tensor $\ten{G}$ is known, however, the factor matrices can be uniquely identified up to trivial scaling ambiguities under suitable conditions \cite{favier2014overview}. By contrast, the PARAFAC decomposition is essentially unique (up to permutation and scaling ambiguities) under mild identifiability conditions. For more details on Tucker and PARAFAC models, we refer the reader to \cite{Sidiropoulos2017,Kolda2009}.

The latter preliminaries will be instrumental in the sequel, where the SIM-assisted MIMO communication model is reformulated and analyzed from a tensor perspective.

\section{System Model and Channel Training} 
We consider a SIM-assisted MIMO system in which a Tx with $M_T$ antennas communicates with a Rx with $M_R$ antennas via a SIM, as depicted in Fig.~\ref{fig:system_model}. The SIM is composed of $L$ parallel metasurface layers, each containing $N$ meta-atoms (elements), and each layer is controlled by a programmable phase-shift matrix. During training, the direct Tx-Rx path is assumed to be negligible (or removed by calibration), so the dominant propagation occurs through the cascaded Tx-SIM-Rx link. Let $\mathcal M\triangleq\{1,\dots,M_T\}$, $\mathcal R\triangleq\{1,\dots,M_R\}$, $\mathcal L\triangleq\{1,\dots,L\}$, and $\mathcal N\triangleq\{1,\dots,N\}$ denote the transmit-antenna, receive-antenna, layer, and meta-atom index sets. The SIM geometry follows a stacked architecture: adjacent layers are separated by a distance $d_{\mathrm{layer}}$, meta-atoms are arranged with inter-element spacing $d_{\mathrm{atom}}$, and all layers are assumed to be co-centered and parallel. The distance between the Tx array and the first SIM layer is denoted by $d_{\mathrm{tx}}$, and the distance between the last SIM layer and the Rx array is represented by $d_{\mathrm{rx}}$.

\revblue{The adopted system model is intended for controlled or semi-controlled SIM-assisted links in which the SIM geometry and inter-layer coupling can be calibrated. Representative deployments include indoor hotspot links, holographic MIMO, access-point-integrated SIM panels, and environmental SIM infrastructure to assist with blocked links~\cite{SIM_CST,HMIMO}. 
In such settings, the SIM may be integrated with the Tx, placed as an intermediate passive or semipassive structure between the Tx and Rx, or installed on room walls or ceilings as intelligent infrastructure for controlled wave propagation.}

For each $\ell$-th SIM layer, the transmission coefficient matrix is given as follows:
\begin{equation}
\begin{aligned}
\mathbf \Phi_{\ell}&\triangleq\mathrm{Diag}(\boldsymbol\phi_{\ell})\in\mathbb C^{N\times N},\\
\boldsymbol\phi_{\ell} &\triangleq \revblue{\left[e^{j\theta_{\ell,1}},\dots,e^{j\theta_{\ell,N}}\right]^T}\in\mathbb C^{N\times 1},
\end{aligned}
\end{equation}
where $\theta_{\ell,n}\in[0,2\pi)$ and $|e^{j\theta_{\ell,n}}|=1$ for ideal lossless meta-atoms (or amplitude-limited values under practical hardware constraints). The EM coupling between two adjacent layers is modeled by $\mathbf W_\ell\in\mathbb C^{N\times N}$, $\forall \ell=2,\dots,L$. Following the Rayleigh--Sommerfeld diffraction-based modeling, each $(n,n')$-th entry of $\mathbf W_\ell$ can be expressed as follows~\cite{SIM_paper}:
\begin{equation}
\label{eq:rayleigh_sommerfeld}
\begin{aligned}
{[\mathbf W_\ell]_{n,n'}}
&\triangleq
\frac{d_x d_y d_{\mathrm{layer}}}{(d_{\ell,n,n'})^2}
\left(\frac{1}{2\pi d_{\ell,n,n'}}-\frac{1}{j\lambda}\right)\\
&\quad\times
\exp\!\left(\frac{j2\pi d_{\ell,n,n'}}{\lambda}\right),\,\,n,n'=1,\ldots,N.
\end{aligned}
\end{equation}
where $\lambda$ is the wavelength and $d_{\ell,n,n'}$ is the distance between meta-atom $n'$ of layer $\ell-1$ and meta-atom $n$ of layer $\ell$. Accordingly, the SIM cascade is the ordered matrix product
\begin{equation}\label{eq:SIMcascade}
\begin{aligned}
\mathbf S &\triangleq \mathbf \Phi_{L}\mathbf W_L\mathbf \Phi_{L-1}\mathbf W_{L-1}\cdots\mathbf \Phi_{2}\mathbf W_2\mathbf \Phi_{1}.
\end{aligned}
\end{equation}
Note that this structure is analogous to a deep linear network in the wave domain~\cite{Semantic_Alignment, SIM_MINN2026}: each programmable layer applies element-wise complex weighting, while each $\mathbf W_\ell$ mixes waves between adjacent layers.

\revblue{It is worth noting that the input-output signal model discussed here is formulated for a pilot-aided training phase in a single transmission direction. The estimated factors are therefore used to reconstruct the SIM-assisted cascaded channel for that directed link. If the system operates in a calibrated time-division duplexing (TDD) mode, the reciprocal-link CSI can be obtained by reversing the order of the SIM cascade and transposing the estimated physical or effective factors. In the odd-layer case, this mapping applies to the effective sub-cascade factors, since the individual fixed inter-layer transfer matrices are absorbed into $\mathbf Z_G$ and $\mathbf Z_H$. In the even-layer case, when the stated identifiability conditions hold, the physical channel factors $\mathbf G$ and $\mathbf H$ can be mapped to the reciprocal link, up to a remaining scalar ambiguity that cancels in cascaded channel reconstruction.}

\begin{figure}[t]
\centering
\includegraphics[width=\columnwidth]{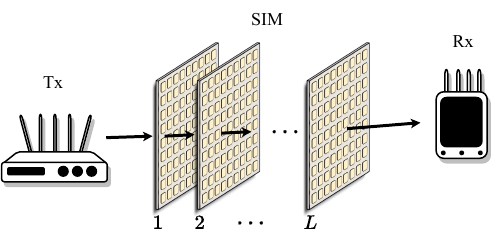}
\caption{The considered SIM-assisted MIMO communication system comprising a Tx with $M_T$ antennas, an Rx with $M_R$ antennas, and a SIM with $L$ metasurface layers, each consisting of $N$ response-tunable metamaterials.}
\label{fig:system_model}
\end{figure}

\subsection{Proposed Training Protocol}
We assume a pilot-based training stage under a two-timescale protocol. The pilot structure comprises $K$ blocks of $T$ time slots, with the pilot sequences repeated across the $K$ blocks. The SIM phase configuration is fixed within a time slot and varies from block to block. Adopting a quasi-static block-fading model, the physical propagation channels remain invariant within the training interval. In particular, $\mathbf G$, $\mathbf H$, and the inter-layer coupling matrices $\{\mathbf W_\ell\}_{\ell=2}^L$ are treated as constant during estimation, while the Rx has full knowledge of the pilot matrix $\mathbf X$. For each pilot block $k\in\{1,\dots,K\}$, the received signal is modeled as follows:
\begin{equation}
\mathbf Y_k \revblue{=} \mathbf G\,\mathbf S_k\,\mathbf H\,\mathbf X + \mathbf B_k,
\end{equation}
where $\mathbf Y_k\in\mathbb C^{M_R\times T}$ is the received pilot matrix, $\mathbf X\in\mathbb C^{M_T\times T}$ is the known pilot matrix, $\mathbf H\in\mathbb C^{N\times M_T}$ is the Tx-to-first-layer channel, $\mathbf G\in\mathbb C^{M_R\times N}$ is the last-layer-to-Rx channel, and $\mathbf B_k\in\mathbb C^{M_R\times T}$ is the additive white Gaussian noise term.
After matched filtering with $\mathbf X^H$, we define:
\begin{equation}\label{eq:Zk}
\mathbf Z_k \triangleq \mathbf Y_k\mathbf X^H
=
\mathbf G\,\mathbf S_k\,\mathbf H + \widetilde{\mathbf B}_k,
\,\,
\widetilde{\mathbf B}_k=\mathbf B_k\mathbf X^H.
\end{equation}
Thus, vectorizing each block gives:
\begin{equation}
\mathbf z_k\triangleq \mathrm{vec}(\mathbf Z_k)
=
(\mathbf H^T\!\otimes\!\mathbf G)\,\mathrm{vec}(\mathbf S_k)+\mathbf b_k,
\,\,
\mathbf b_k\triangleq \mathrm{vec}(\widetilde{\mathbf B}_k).\nonumber
\end{equation}
Collecting all blocks at the Rx yields:
\begin{equation}
\mathbf Z \triangleq [\mathbf z_1,\dots,\mathbf z_K]
=
(\mathbf H^T\!\otimes\!\mathbf G)\,\mathbf S^T + \mathbf B,\label{eq:eqsimcascade}
\end{equation}
where $\mathbf S^T\triangleq[\mathrm{vec}(\mathbf S_1),\dots,\mathrm{vec}(\mathbf S_K)]$ and
$\mathbf B\triangleq[\mathbf b_1,\dots,\mathbf b_K]$.

\revblue{We now summarize the main assumptions used by the proposed estimators as follows. The pilot matrix $\mathbf X$ is known at the Rx and is chosen so that the matched-filtered observations in \eqref{eq:Zk} are available. The Tx-SIM channel $\mathbf H$, the SIM-Rx channel $\mathbf G$, and the inter-layer coupling matrices $\{\mathbf W_\ell\}_{\ell=2}^{L}$ are quasi-static during the training interval. The direct Tx-Rx path is negligible or has been removed through calibration. The nominal SIM training configurations are known to the Rx, while the blind/imperfect-training variant in Section~\ref{sec:blind-SIM-variant} relaxes this assumption by jointly updating the training factor. Also, the inter-layer coupling matrices of the SIM are assumed known from geometry-based modeling or tailored electromagnetic calibration. 
Without loss of generality, Rayleigh fading channels and ideal unit-modulus phase shifts are adopted in our simulations.}

\subsection{Least Squares (LS) estimator}\label{subsec:3b}
The equivalent input-output SIM cascade model in (\ref{eq:eqsimcascade}) also suggests a simple unstructured least-squares (LS) reference estimator. Let us define the composite channel matrix
\begin{equation}
\mathbf T \triangleq \mathbf H^T\!\otimes\!\mathbf G \in \mathbb C^{M_RM_T\times N^2},
\end{equation}
so that $\mathbf Z=\mathbf T\mathbf S^T+\mathbf B$. If the SIM cascade training matrix $\mathbf S^T$ is known from the programmed phase configurations and has sufficient rank, $\mathbf T$ can be estimated directly as follows:
\begin{equation}
\widehat{\mathbf T}_{\mathrm{LS}}
=
\argmin_{\mathbf T}\left\|\mathbf Z-\mathbf T\mathbf S^T\right\|_F^2
=
\mathbf Z\left(\mathbf S^T\right)^+.\label{eq:aggLS}
\end{equation}
This estimator treats $\mathbf H^T\!\otimes\!\mathbf G$ as a single unconstrained matrix and, therefore, does not exploit the Kronecker structure of the composite channel factor. It also uses only the aggregate training matrix $\mathbf S^T$, and does not exploit the odd/even SIM cascade structures developed in the proposed TenSIM framework. For this reason, the unstructured LS estimator will later serve as a reference baseline in the numerical comparisons.

A major limitation of this aggregate LS baseline is its training overhead. Since $\mathbf T$ is estimated as a fully unstructured $M_RM_T\times N^2$ matrix, the training matrix $\mathbf S^T\in\mathbb C^{N^2\times K}$ must have full row rank for the LS estimate to be identifiable in the noiseless case. This requires at least $K\geq N^2$ pilot blocks, which quickly becomes impractical for large SIM apertures. When $K<N^2$, the aggregate LS problem is underdetermined unless additional regularization or prior information is introduced. In contrast, as will be clear later, the TenSIM estimators exploit the PARAFAC/Tucker structure induced by the SIM cascade and the central-layer training design, thereby working with a much smaller effective training dimension $K$ rather than estimating $\mathbf T$ as a completely unstructured matrix.

\begin{figure*}[!t]
\centering
\includegraphics[width=1.0\textwidth]{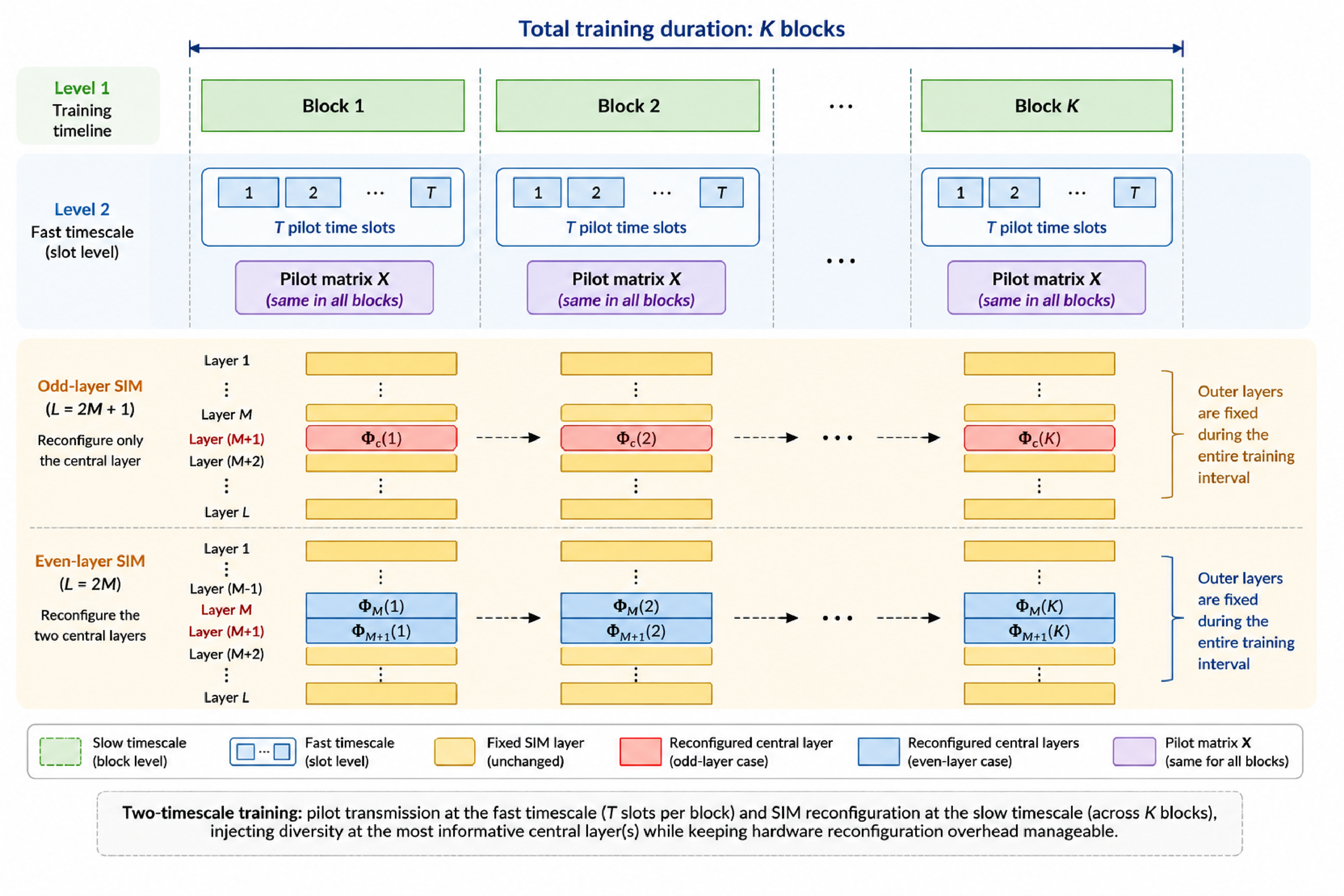}
\caption{Proposed two-timescale SIM training protocol.}
\label{fig:training_protocol}
\end{figure*}

\subsection{Proposed SIM Training Strategies} 
The proposed two-timescale SIM training protocol is illustrated in Fig.~\ref{fig:training_protocol}. From a practical standpoint, reconfiguring all $L$ metasurface layers for each pilot block is undesirable. Such a strategy increases control signaling overhead~\cite{RIS_Control}, requires tighter synchronization between the controller and transceivers, increases calibration complexity, and can exceed hardware switching-speed limits in systems with large SIM. Therefore, the training protocol should introduce sufficient diversity to ensure channel identifiability, while minimizing the number of coefficients that need to be updated. This leads to a reduced-complexity strategy in which only the most informative central layers are varied across the $K$ pilot blocks, while the remaining layers are kept fixed. We consider two cases: an odd number of layers ($L=2M+1$) and an even number of layers ($L=2M$). 
\begin{itemize}
\item For $L=2M+1$: training is assigned to the middle (center) layer $M+1$, which means that only the phase shifts of the middle layer are updated during training, whereas the remaining $2M$ layers ($1 \rightarrow M$) and  ($M+2 \rightarrow 2M+1$)  are kept fixed.
\item For $L=2M$: training is assigned to the two middle layers $M$ and $M+1$, which means that the phase shifts of these two middle layers are jointly updated during training, whereas the remaining $2M-2$ layers ($1 \rightarrow M-1$) and  ($M+2 \rightarrow 2M$) are kept fixed.
\end{itemize}

\revblue{\noindent\textit{Remark 1}: This two-timescale design should be viewed as a structured and hardware-aware training rule rather than as a globally optimal design for every impairment model. It is expected that phase noise or phase quantization may perturb the intended training factors, while switching delay and layer-dependent reconfiguration latency limit the number of iterations of the proposed ALS algorithm, which motivates varying only the central layer(s). 
The imperfect/blind training extension in Section~\ref{sec:blind-SIM-variant} partly addresses such perturbations by jointly updating the unknown training factor, at the cost of additional pilot diversity and increased sensitivity to initialization. It is noted that robust codebook optimization under hardware constraints is outside the scope of this work. }

\section{Tensor-Based Channel Estimation Formulation}
In this section, we first recast the received pilot observations into a third-order tensor model by stacking the matched-filtered blocks along the training-block index. This observation model captures the multilinear structure seen at the receiver and provides a unified starting point for the subsequent developments. Based on this tensorized signal representation, we derive the SIM cascade tensor structures induced by the two training strategies: the odd-layer case ($L=2M+1$) and the even-layer case ($L=2M$).

From expression~(\ref{eq:Zk}), we define the following third-order tensors:
\begin{equation}
\ten Z(:,:,k)\triangleq\mathbf Z_k,\quad
\ten S(:,:,k)\triangleq\mathbf S_k,\quad
\ten B(:,:,k)\triangleq\widetilde{\mathbf B}_k,
\end{equation}
so that $\ten Z\in\mathbb C^{M_R\times M_T\times K}$ and
$\ten S\in\mathbb C^{N\times N\times K}$. Then the received signal model can be equivalently written as follows:
\begin{equation}\label{Ztensor}
\ten Z
=
\ten S\times_1 \mathbf G \times_2 \mathbf H^T + \ten B,
\end{equation}
which is a Tucker-$(2,3)$ form, i.e., two loading matrices ($\mathbf {G}$ and $\mathbf {H} ^T$) act on modes 1 and 2, while mode 3 is preserved. 


\subsection{SIM with Odd Number of Layers ($L=2M+1$)}
For the odd-layer architecture, we update only the single central layer during training, while all remaining SIM layers are kept fixed. Hence, for each pilot block $k\in\{1,\ldots,K\}$, the following holds:
\begin{equation}
\boldsymbol\phi_{\ell,k}=\boldsymbol\phi_{\ell},\;\ell\neq M+1,
\qquad
\boldsymbol\phi_{M+1,k}\ \text{varies with }k.
\end{equation}
From (\ref{eq:SIMcascade}), we have:
\begin{equation}\label{eq:SIMcascadek}
\begin{aligned}
\mathbf S_k &= \mathbf \Phi_{L}\mathbf W_L\mathbf\cdots \mathbf \Phi_{M+1,k}\mathbf W_{M+1}\cdots \mathbf \Phi_{2}\mathbf W_2\mathbf \Phi_{1},
\end{aligned}
\end{equation}
where only layer $M+1$ depends on the training block index $k$. Taking the $\textrm{vec}(\cdot)$ operator, applying multiple times property (\ref{eq:vecABC}), and concatenating for all $K$ blocks, we obtain:
\begin{equation}\label{eq:STodd}
\begin{aligned}
\mathbf S^T
&=
\left[
\prod_{m=1}^{M}
\big(\mathbf \Phi_m\otimes \mathbf \Phi_{2M-m+2}\big)
\big(\mathbf W_{m+1}^T\otimes \mathbf W_{2M-m+2}\big)
\right]\\
&\quad\times \mathbf S_{M+1}.
\end{aligned}
\end{equation}
with $\mathbf S_{M+1}
\triangleq
[\mathrm{vec}(\mathbf \Phi_{M+1,1}),\dots,\mathrm{vec}(\mathbf \Phi_{M+1,K})] \in\mathbb C^{N^2\times K}$.
Note that expression~(\ref{eq:STodd}) represents the SIM cascade pairing from the two sides (first and last layers) toward the center.
Let $\mathbf P\in\mathbb R^{N^2\times N}$ be the diagonal-selection matrix such that the following holds
\begin{equation}
\mathrm{vec}(\mathbf \Phi_{M+1,k})=\mathbf P\,\boldsymbol\phi_{M+1,k},
\qquad
\mathbf S_{M+1}=\mathbf P\,\mathbf \Phi,
\end{equation}
where the following matrix:
\begin{equation}
\mathbf \Phi\triangleq\mathbf \Phi_{M+1}=[\boldsymbol\phi_{M+1,1},\dots,\boldsymbol\phi_{M+1,K}]\in\mathbb C^{N\times K}
\end{equation}
denotes the middle-layer SIM training matrix. Collecting the fixed left/right factors in (\ref{eq:STodd}), by repeatedly using the mixed-product property of the Kronecker product and using the diagonal-selection matrix $\mathbf P$ together with the identity $(\mathbf B\!\otimes\!\mathbf A)\mathbf P=\mathbf B\!\diamond\!\mathbf A$, we obtain
\begin{equation}\label{eq:ST0}
\begin{aligned}
\mathbf S^T
&=
\mathbf T
\left[
\big(\mathbf \Phi_M\mathbf W_{M+1}^T\big)
\diamond
\big(\mathbf \Phi_{M+2}\mathbf W_{M+2}\big)
\right] \mathbf \Phi^T,
\end{aligned}
\end{equation}
where
\begin{equation}
\label{eq:ST}
\begin{aligned}
\mathbf T
&\triangleq\prod_{m=1}^{M-1}
\Big[
\big(\mathbf \Phi_m\mathbf W_{m+1}^T\big)\\
&\qquad\otimes \big(\mathbf \Phi_{2M-m+2}\mathbf W_{2M-m+2}\big)
\Big]
\in\mathbb C^{N^2\times N^2}.
\end{aligned}
\end{equation}
Expressions (\ref{eq:ST0}) and (\ref{eq:ST}) make explicit that all fixed layers are grouped inside a single Kronecker factor $\mathbf T$, while the two middle varying layers appear separately inside the Khatri--Rao product. To further separate the two SIM cascades, we use the following Kronecker product property:
\begin{equation}\label{eq:propertyprodkron}
\prod_{m=1}^{R}\left(\mathbf A_m\otimes\mathbf B_m\right)
=
\left(\prod_{m=1}^{R}\mathbf A_m\right)
\otimes
\left(\prod_{m=1}^{R}\mathbf B_m\right),
\end{equation}
where the correspondences are $\mathbf A_m\triangleq \mathbf \Phi_m\mathbf W_{m+1}^T$ and $\mathbf B_m\triangleq \mathbf \Phi_{2M-m+2}\mathbf W_{2M-m+2}$. This allows us to group the left and right SIM factors and write $\mathbf T=\mathbf T_{\textrm{L}}\otimes \mathbf T_{\textrm{R}}$, where
\begin{align}
\mathbf T_{\textrm{L}} &\triangleq  \prod_{m=1}^{M-1}\mathbf \Phi_m\mathbf W_{m+1}^T \in\mathbb C^{N\times N}\\
&=
\mathbf \Phi_1\mathbf W_2^T\,\mathbf \Phi_2\mathbf W_3^T\cdots\mathbf \Phi_{M-1}\mathbf W_M^T,\nonumber\\
\mathbf T_{\textrm{R}} &\triangleq  \prod_{m=1}^{M-1} \mathbf \Phi_{2M-m+2}\mathbf W_{2M-m+2}\in\mathbb C^{N\times N}\\
&= \mathbf \Phi_{2M+1}\mathbf W_{2M+1}\,\cdots\mathbf \Phi_{M+3}\mathbf W_{M+3}.\nonumber
\end{align}
Physically, $\mathbf T_{\textrm{L}}$ aggregates the wave transformations along the first half of the SIM cascade (from the Tx side up to the neighborhood of the central layers), while $\mathbf T_{\textrm{R}}$ aggregates the complementary transformations along the second half (from the Rx side back toward the center). In other words, $\mathbf T_{\textrm{L}}$ and $\mathbf T_{\textrm{R}}$ are two directional sub-cascades that collect all fixed-layer effects outside the block-varying middle layers. This Kronecker-based split of the full cascade into the two structured parts $\mathbf T_{\textrm{L}}$ and $\mathbf T_{\textrm{R}}$ induces a consistent PARAFAC structure for the received pilot tensor, which is then exploited to perform decoupled estimation of the Tx-SIM and SIM-Rx channel factors.

Since $\mathbf T=\mathbf T_{\textrm{L}}\otimes \mathbf T_{\textrm{R}}$, we can rewrite (\ref{eq:ST0}) as follows:
\begin{equation}\label{eq:STcompact}
\mathbf S^T
= (\mathbf S_{\textrm{L}} \diamond \mathbf S_{\textrm{R}})\mathbf \Phi^T,
\end{equation}
where $\mathbf S_{\textrm{L}}\triangleq\mathbf{T}_{\textrm{L}}\mathbf \Phi_M\mathbf W_{M+1}^T \in \mathbb{C}^{N \times N}$ and $\mathbf S_{\textrm{R}}\triangleq\mathbf T_{\textrm{R}}\mathbf \Phi_{M+2}\mathbf W_{M+2} \in \mathbb{C}^{N \times N}$.
We can interpret $\mathbf S^T$ as the transposed mode-3 unfolding of a tensor $\ten S \in\mathbb C^{N\times N\times K}$ that follows a rank-$N$ PARAFAC decomposition given by:
\begin{equation}\label{eq:Stensor0}
\ten S
=
\ten I_{3,N}
\times_1 \mathbf S_{\textrm{R}}
\times_2
 \mathbf S_{\textrm{L}}
\times_3
\mathbf \Phi_{}.
\end{equation}


Finally, substituting the PARAFAC decomposition of the SIM cascade in (\ref{eq:STcompact}) into the received signal tensor model (\ref{Ztensor}), and using the property of the $n$-mode product $\ten X \times_n \mathbf A \times_n \mathbf B = \ten X \times_n (\mathbf B_n\mathbf A_n)$, $n=1$ and $2$, we obtain:
\begin{equation}\label{eq:ZtensorCP}
\ten Z
=
\ten I_{3,N}
\times_1 \mathbf G\mathbf S_{\textrm{R}}
\times_2 \mathbf H^T \mathbf S_{\textrm{L}}
\times_3 \mathbf \Phi_{} + \ten{B},
\end{equation}
or, equivalently,
\begin{equation}\label{eq:oddlayer_parafac_model}
\ten Z
=
\ten I_{3,N}
\times_1 \mathbf Z_G
\times_2 \mathbf Z_H
\times_3 \mathbf \Phi_{} + \ten B
\end{equation}
where
\begin{eqnarray}\label{eq:ZgZh_odd}
\mathbf Z_G &\triangleq& \mathbf G\mathbf T_{\textrm{R}}\mathbf \Phi_{M+2}\mathbf W_{M+2},\nonumber\\
\mathbf Z_H &\triangleq& \mathbf H^T\mathbf T_{\textrm{L}}\mathbf \Phi_M\mathbf W_{M+1}^T.
\end{eqnarray}
\revblue{Thus, in the odd-layer PARAFAC model, the directly identifiable channel-related quantities are the effective factors $\mathbf Z_G$ and $\mathbf Z_H$. These factors absorb the fixed SIM sub-cascades surrounding the trained middle layer. Consequently, the physical matrices $\mathbf G$ and $\mathbf H$ cannot be recovered separately from this model, unless additional calibration or side information about the fixed sub-cascades is available. Nevertheless, the end-to-end cascaded channel can be reconstructed from $\mathbf Z_G$, $\mathbf Z_H$, and the tested middle-layer phase profile.}



For the odd-layer branch of TenSIM, the proposed SIM training structure induces the PARAFAC model in (\ref{eq:ZtensorCP}). Hence, SIM channel estimation can be rigorously cast as the identification of a structured PARAFAC model from the received pilot tensor, with factors $\mathbf Z_G$, $\mathbf Z_H$, and $\mathbf \Phi_{}$. This formulation enables decoupled estimation of the cascaded Tx-SIM and SIM-Rx channel components, while preserving identifiability under mild conditions via the essential uniqueness property of PARAFAC. This is a distinguishing feature of the proposed TenSIM approach that has not yet been addressed in prior relevant work. 

\subsection{SIM with Even Number of Layers ($L=2M$)}
For the even-layer architecture, we adopt the same two-timescale training logic as in Section~3.1, adapted to the absence of a single central layer. Specifically, when $L=2M$, the two middle layers are numbered as $M$ and $M+1$, and only these two layers are reconfigured across pilot blocks. All remaining layers are kept fixed during the full training interval. Similar to the odd-layer case, this design preserves blockwise diversity in the cascaded SIM response while keeping the number of time-varying phase shifts controlled. Hence, for each pilot block $k\in\{1,\ldots,K\}$, we have:
\begin{eqnarray}
&&\boldsymbol\phi_{\ell,k}=\boldsymbol\phi_\ell,\quad \ell\notin\{M,M+1\},\nonumber\\
&&\boldsymbol\phi_{M,k},\;\boldsymbol\phi_{M+1,k}\ \text{vary with }k.
\end{eqnarray}
Then, the block-dependent cascade can be written as follows:
\begin{align}\label{eq:SIMcascadek2}
\mathbf S_k =&\mathbf \Phi_{L}\mathbf W_L\mathbf \cdots \mathbf \Phi_{M+1,k}\mathbf W_{M+1}\mathbf \Phi_{M,k} \cdots\mathbf W_2\mathbf \Phi_{1},
\end{align}
where the two central layers $M$ and $M+1$ are used for training, depending on the block index $k$. Note that since $L=2M$, the center of the SIM cascade corresponds to the central propagation matrix $\mathbf {W}_{M+1}$. Applying repeatedly property (\ref{eq:vecABC}) to the SIM cascade (\ref{eq:SIMcascadek2}) yields
\begin{equation}
\mathbf s_k \triangleq \mathrm{vec}(\mathbf S_k)
=
\mathbf T
\left(\boldsymbol{\Phi}_{M,k} \otimes \boldsymbol{\Phi}_{M+1,k}\right)
\mathrm{vec}(\mathbf W_{M+1}),
\end{equation}
where the matrix $\mathbf T$ now collects the contributions of all outer SIM layers, and can be written as follows:
\begin{equation}\label{eq:Teven}
\mathbf T =
\prod_{m=1}^{M}
\big[
(\bm{\Phi}_m \otimes \bm{\Phi}_{2M-m+1})
(\mathbf W_{m+1}^T \otimes \mathbf W_{2M-m+1})
\big].
\end{equation}
Noting that
\begin{equation}
\begin{aligned}
\boldsymbol{\Phi}_{M,k} \otimes \boldsymbol{\Phi}_{M+1,k}
&=
\textrm{Diag}(\bm{\phi}_{M,k}) \otimes \textrm{Diag}(\bm{\phi}_{M+1,k})\\
&=
\textrm{Diag}(\bm{\phi}_{M,k} \otimes \bm{\phi}_{M+1,k}),
\end{aligned}
\end{equation}
and defining $\bar{\bm{\phi}}_k=\bm{\phi}_{M,k} \otimes \bm{\phi}_{M+1,k}$
as the combined phase shift vector for the two central layers, the SIM cascade becomes:
\begin{equation}
\mathbf s_k
=
\mathbf T
\textrm{Diag}(\bar{\boldsymbol{\phi}}_k)
\mathrm{vec}(\mathbf W_{M+1}).
\end{equation}
Finally, stacking the $K$ blocks yields:
\begin{equation}
\mathbf S^T
\!\!=\!\!
\mathbf T
\left[
\textrm{Diag}(\bar{\boldsymbol{\phi}}_1)\mathrm{vec}(\mathbf W_{M+1}),
\ldots,
\textrm{Diag}(\bar{\boldsymbol{\phi}}_K)\mathrm{vec}(\mathbf W_{M+1})
\right]\!\!.
\end{equation}
or, equivalently,
\begin{equation}\label{eq:STeven}
\mathbf S^T
=
\mathbf T
\textrm{Diag}(\mathrm{vec}(\mathbf W_{M+1}))
\bar{\boldsymbol{\Phi}}^T,
\end{equation}
where the following matrix:
\begin{equation}
\bar{\boldsymbol{\Phi}}= [\bar{\boldsymbol{\phi}}_1,\dots,\bar{\boldsymbol{\phi}}_K]^T= (\boldsymbol{\Phi}_{M} \diamond \boldsymbol{\Phi}_{M+1})^T 
\in \mathbb{C}^{K\times N^2}
\end{equation}
defines the overall SIM training matrix that combines the two central layers ($M$ and $M+1$), and
$\boldsymbol{\Phi}_{M}=[\bm{\phi}_{M,1},\ldots,  \bm{\phi}_{M,K}] \in \mathbb{C}^{N \times K}$ and $\boldsymbol{\Phi}_{M+1}=[\bm{\phi}_{M+1,1},\ldots,  \bm{\phi}_{M+1,K}] \in \mathbb{C}^{N \times K}$ are the training matrices of these layers. In a procedure similar to the odd-SIM case, applying property (\ref{eq:propertyprodkron}) to (\ref{eq:Teven}), we group left and right SIM factors and write $\mathbf T=\mathbf T_{\textrm{L}}\otimes \mathbf T_{\textrm{R}}$, where
\begin{align}
\mathbf T_{\textrm{L}} &=  \prod_{m=1}^{M-1}\mathbf \Phi_m\mathbf W_{m+1}^T\\
&=
\mathbf \Phi_1\mathbf W_2^T\,\mathbf \Phi_2\mathbf W_3^T\cdots\mathbf \Phi_{M-1}\mathbf W_M^T,\nonumber\\
\mathbf T_{\textrm{R}} &=  \prod_{m=1}^{M-1} \mathbf \Phi_{2M-m+1}\mathbf W_{2M-m+1}\\
&= \mathbf \Phi_{2M}\mathbf W_{2M}\,\cdots\mathbf \Phi_{M+2}\mathbf W_{M+2}.\nonumber
\end{align}
This allows us to rewrite (\ref{eq:STeven}) as follows:
\begin{equation}\label{eq:STeven2}
\mathbf S^T
=
(\mathbf T_{\textrm{L}}\otimes \mathbf T_{\textrm{R}})
\textrm{Diag}(\mathrm{vec}(\mathbf W_{M+1}))
(\boldsymbol{\Phi}_{M} \diamond \boldsymbol{\Phi}_{M+1}).
\end{equation}
It can be concluded that $\mathbf S^T$ is the transposed mode-3 unfolding of a tensor $\ten S \in\mathbb C^{N\times N\times K}$ that follows a  Tucker decomposition given by:
\begin{equation}\label{eq:StensorTucker}
\ten S
=
\ten W_{M+1}
\times_1 \mathbf T_{\textrm{R}}
\times_2
\mathbf T_{\textrm{L}}
\times_3
(\boldsymbol{\Phi}_{M} \diamond \boldsymbol{\Phi}_{M+1})^T,
\end{equation}
where $\ten W_{M+1} \in\mathbb C^{N\times N\times N}$ is the core tensor such that $[\ten W_{M+1}]_{(3)}= \textrm{Diag}(\mathrm{vec}(\mathbf W_{M+1}))$. 
By substituting the SIM cascade (\ref{eq:StensorTucker}) into the received signal model (\ref{Ztensor}), we get
\begin{equation}
\ten Z
=
\ten W_{M+1} 
\times_1 \mathbf Z_G
\times_2 \mathbf Z_H
\times_3 \bar{\mathbf \Phi}_{} \\
\in\mathbb C^{N\times N \times K},\label{eq:ZtensorCP_even}
\end{equation}
where
\begin{equation}\label{eq:ZgZh_even}
\mathbf Z_G = \mathbf G\mathbf T_{\textrm{R}},
\qquad
\mathbf Z_H = \mathbf H^T\mathbf T_{\textrm{L}}.
\end{equation}
\revblue{In the even-layer Tucker model, the known core $\ten W_{M+1}$ preserves the coupling induced by the central inter-layer propagation matrix. Under the rank conditions discussed later and when the fixed sub-cascades are known or calibrated, the ALS updates can be interpreted as estimating the physical channel factors $\mathbf G$ and $\mathbf H$ through their products with $\mathbf T_{\textrm{R}}$ and $\mathbf T_{\textrm{L}}$, up to the intrinsic scalar ambiguity. Similar to the odd-layer case, the ambiguity cancels when reconstructing the end-to-end cascaded channel for a given SIM configuration.}

\subsection{Odd vs. Even Tensor Constructions}
We will now compare the two proposed TenSIM tensor constructions. On the SIM-cascade side, the odd-layer model in (\ref{eq:STcompact}) yields a third-order tensor with PARAFAC structure, whereas the even-layer model in (\ref{eq:StensorTucker}) yields a third-order tensor with Tucker structure. In both cases, the SIM training process is represented by a third-order tensor whose first two modes correspond to the two directional SIM sub-cascades, and whose third mode captures the block-varying training factors. The key difference is that, for odd $L$, the presence of a single varying middle layer leads to a superdiagonal core and, therefore, to a PARAFAC model, while, for even $L$, the two varying middle layers and the intermediate propagation matrix $\mathbf W_{M+1}$ generate a nontrivial core tensor $\mathcal W_{M+1}$, which naturally leads to a Tucker model.

The same correspondence appears in the received pilot tensors. The odd-layer received model in (\ref{eq:ZtensorCP}) preserves the PARAFAC form with factors $\mathbf Z_G$, $\mathbf Z_H$, and $\mathbf \Phi$, whereas the even-layer received model in (\ref{eq:ZtensorCP_even}) preserves the Tucker form with factor matrices $\mathbf Z_G$, $\mathbf Z_H$, and $\bar{\mathbf \Phi}$ coupled through the core $\mathcal W_{M+1}$. The TenSIM distinction between PARAFAC for odd-layer SIMs and Tucker for even-layer SIMs is a direct consequence of the different physical roles of the central part of the SIM cascade under the proposed training strategy.
\revblue{This distinction is therefore broader than a purely algebraic rearrangement, but it is also tied to the chosen low-reconfiguration protocol. To this end, with a single trainable central layer, the middle of the cascade acts as a diagonal modulation, separating two effective directional sub-cascades. With two trainable central layers, the field must propagate through the central inter-layer matrix, so the coupling cannot be diagonalized away and appears as a Tucker core. Training protocols that excite additional layers could yield richer coupled tensor models, whereas the proposed protocol selects the central layer(s) to obtain the simplest identifiable structure compatible with the parity of the stack.}

\subsection{Design of SIM Training Configuration}
\label{sec:training_design}
Now that the odd- and even-layer training factors $\mathbf \Phi$ and $\bar{\mathbf \Phi}$ have been explicitly defined, we discuss how to design these SIM phase profiles in practice. Beyond identifiability, the training matrices affect the numerical conditioning of the LS updates and must remain compatible with practical hardware constraints. In particular, useful training designs should provide sufficient diversity across blocks, low mutual correlation among phase patterns, and unit-modulus entries implementable by the metasurface~\cite{RIS_survey}.

Accordingly, let $\mathbf \Phi_{\ell,k}=\mathrm{Diag}(\boldsymbol{\phi}_{\ell,k})$ denote the diagonal phase-shift matrix applied at layer $\ell$ during block $k$, where $\boldsymbol{\phi}_{\ell,k}\in\mathbb C^N$ collects the corresponding reflection coefficients. For practical SIM hardware, these coefficients satisfy $|[\boldsymbol{\phi}_{\ell,k}]_n|=1$, $\forall n$. For odd-layer SIMs, only the central layer is varied during training, so the resulting observation tensor is driven by a single training matrix
\begin{equation}
\mathbf \Phi
=
\begin{bmatrix}
\boldsymbol{\phi}(1) & \boldsymbol{\phi}(2) & \cdots & \boldsymbol{\phi}(K)
\end{bmatrix}
\in\mathbb C^{N\times K}.
\end{equation}
A natural choice is a structured unit-modulus design, such as truncated discrete Fourier transform (DFT) columns~\cite{Tellambura1999DFT} or CAZAC-type sequences~\cite{Chu1972}, which promote low correlation between training blocks and improve the conditioning of the associated estimation problems. In particular, when $K\leq N$, choosing distinct columns of a normalized DFT matrix provides orthogonal training patterns while preserving constant modulus. For even-layer SIMs, the two central layers are jointly reconfigured, and the effective training factor is the following:
\begin{equation}
\bar{\mathbf \Phi}
=
(\boldsymbol{\Phi}_{M} \diamond \boldsymbol{\Phi}_{M+1})^T
\in \mathbb{C}^{K \times N^2},
\end{equation}
where $\boldsymbol{\Phi}_{M},\boldsymbol{\Phi}_{M+1} \in \mathbb{C}^{N \times K}$.
Thus, the quality of $\bar{\mathbf \Phi}$ is determined by the pair $(\boldsymbol{\Phi}_{M},\boldsymbol{\Phi}_{M+1})$. Note that $\bar{\mathbf \Phi}\,\bar{\mathbf \Phi}^{H}
=
(\boldsymbol{\Phi}_{M}^{H}\boldsymbol{\Phi}_{M})
\odot
(\boldsymbol{\Phi}_{M+1}^{H}\boldsymbol{\Phi}_{M+1})$,
so $\bar{\mathbf \Phi}$ has orthogonal rows if both factors have orthogonal columns, and full row rank if both factors have full column rank; a convenient construction is to choose $\boldsymbol{\Phi}_{M}$ and $\boldsymbol{\Phi}_{M+1}$ from normalized DFT codebooks. Let $[\boldsymbol{\phi}_{M,k}]_n = e^{-j\frac{2\pi}{N}(n-1)p_k}$ and $[\boldsymbol{\phi}_{M+1,k}]_n = e^{-j\frac{2\pi}{N}(n-1)q_k}$, $\forall n=1,\ldots,N$, and enumerate the index pairs as follows:
\begin{equation}
p_k = \left\lfloor \frac{k-1}{N} \right\rfloor,
\qquad
q_k = (k-1) \bmod N,
\qquad k=1,\ldots,K,\nonumber
\end{equation}
with $K\leq N^2$. Then $\bar{\boldsymbol{\phi}}_k
=
\boldsymbol{\phi}_{M,k}\otimes \boldsymbol{\phi}_{M+1,k}$
Equivalently, $\bar{\mathbf \Phi}^T$ consists of distinct columns of the separable two-dimensional DFT codebook $(\mathbf F_N\otimes \mathbf F_N)$. In particular, when $K=N^2$, $\bar{\mathbf \Phi}^T$ coincides with $(\mathbf F_N\otimes \mathbf F_N)$ up to column permutation. 

Overall, the odd-layer case requires one well-conditioned SIM training matrix, whereas the even-layer case benefits from a coordinated design of two central-layer training matrices. These considerations will be revisited later on in the identifiability and numerical evaluation sections. 

\noindent\textit{Remark 2}: For the odd-layer model, choosing $\mathbf \Phi$ from a normalized DFT codebook is also beneficial for the numerical behavior of the PARAFAC-ALS updates developed later. Indeed, using the Gram identity of the Khatri--Rao product $(\mathbf A\diamond \mathbf B)^H(\mathbf A\diamond \mathbf B)=(\mathbf A^H\mathbf A)\odot(\mathbf B^H\mathbf B)$,
if the training matrix satisfies $\mathbf \Phi^H\mathbf \Phi=\mathbf I$ (or, more generally, is column-orthogonal up to a scale factor), then, $(\mathbf \Phi\diamond \mathbf Z_H)^H(\mathbf \Phi\diamond \mathbf Z_H)=\mathbf I\odot(\mathbf Z_H^H\mathbf Z_H)$
and similarly for $\mathbf \Phi\diamond \mathbf Z_G$. 
An analogous expression holds for $\mathbf \Phi\diamond \mathbf Z_G$. Therefore, the Khatri--Rao design matrices appearing in the odd-layer LS subproblems in (\ref{eq:JG}) and (\ref{eq:JH}) and in the corresponding updates (\ref{eq:LS_G_odd}) and (\ref{eq:LS_H_odd}) inherit zero inter-column correlation from the training factor. This fact improves conditioning and simplifies the associated pseudo-inverse computations. In this sense, DFT-based choices of $\mathbf \Phi$ are attractive from a numerical-stability viewpoint.

\section{The Proposed TenSIM Framework}
Having obtained a PARAFAC model for odd-layer SIMs and a Tucker model for even-layer SIMs, we propose estimating the unknown factors using the classical ALS procedure \cite{Comon2009ALS}. In both cases, ALS follows the same principle: at each step, we keep all factors but one fixed, so the tensor model becomes linear in the selected factor, and its update is obtained via an LS solution. This process is repeated factor-by-factor until convergence (or until a maximum number of iterations is reached). The ALS routine shared by both models consists of the steps: \textit{i}) initialize the unknown factors; \textit{ii}) update one factor from the current residual; \textit{iii}) cycle through all remaining factors; and \textit{iv}) evaluate a stopping criterion, e.g., a decrease in relative reconstruction error below a threshold. 

It is noted that the odd and even cases differ in the algebraic form of each LS subproblem (PARAFAC factors in the odd case \textit{versus} Tucker factors in the even case). Accordingly, we adopt two tensor-based algorithmic routes: one tailored to the odd-layer case and the other to the even-layer case, since the SIM training structure differs between them. In what follows, we provide the corresponding optimization problems and LS updates for each case.

\subsection{Odd-Layer Case ($L=2M+1$): PARAFAC-ALS}\label{eq:oddparafac}
From the received signal tensor in (\ref{eq:ZtensorCP}), its $1$- and $2$-mode matrix unfoldings can be written as follows:
\begin{align}
\mathbf Z_{(1)} &= \mathbf G\mathbf S_{\textrm{R}}\,(\mathbf \Phi \diamond \mathbf Z_H)^T + \mathbf B_{(1)},\\
\mathbf Z_{(2)} &= \mathbf H^T\mathbf S_{\textrm{L}}\,(\mathbf \Phi \diamond \mathbf Z_G)^T + \mathbf B_{(2)},
\end{align}
where $\mathbf B_{(1)}$ and $\mathbf B_{(2)}$ are the mode-1 and mode-2 unfoldings of the noise tensor $\ten B$.
From these unfoldings, the following problems are solved:
\begin{align}
J_G(\mathbf G) &\triangleq \left\|\mathbf Z_{(1)}-\mathbf G\mathbf S_{\textrm{R}}(\mathbf \Phi \diamond \mathbf Z_H)^T\right\|_F^2,\label{eq:JG}\\
J_H(\mathbf H) &\triangleq \left\|\mathbf Z_{(2)}-\mathbf H^T\mathbf S_{\textrm{L}}(\mathbf \Phi \diamond \mathbf Z_G)^T\right\|_F^2,\label{eq:JH}
\end{align}
with the corresponding LS minimizers constructed as:
\begin{align}
\widehat{\mathbf G} &= \mathbf Z_{(1)}\Big(\mathbf S_{\textrm{R}}(\widehat{\mathbf \Phi}\diamond (\widehat{\mathbf H}^T\mathbf S_{\textrm{L}}))^T\Big)^+,\label{eq:LS_G_odd}\\
\widehat{\mathbf H}^T &= \mathbf Z_{(2)}\Big(\mathbf S_{\textrm{R}}(\widehat{\mathbf \Phi}\diamond(\widehat{\mathbf G}\mathbf S_{\textrm{L}}))^T\Big)^+,\label{eq:LS_H_odd}
\end{align}
Note that, when $\mathbf \Phi$ is fixed and known to the Rx (usual assumption), only the first two updates are required.

\subsection{Even-Layer Case ($L=2M$): Tucker-ALS}

From (\ref{eq:ZtensorCP_even}), define $\bar{\mathbf \Phi}\triangleq(\boldsymbol{\Phi}_{M} \diamond \boldsymbol{\Phi}_{M+1})^T$ and $\ten W\triangleq\ten W_{M+1}$. The 1-mode and 2-mode unfoldings of the received pilot tensor $\ten{Z}$ can be written as follows:
\begin{align}
\mathbf{Z}_{(1)}
&=
\mathbf{G}\mathbf{T}_{\textrm{R}}
\mathbf{W}_{(1)}
(\bar{\boldsymbol{\Phi}} \otimes \mathbf{Z}_H)^T
+
\mathbf B_{(1)},
\\
\mathbf{Z}_{(2)}
&=
\mathbf{H}^T\mathbf{T}_{\textrm{L}}
\mathbf{W}_{(2)}
(\bar{\boldsymbol{\Phi}} \otimes \mathbf{Z}_G)^T
+
\mathbf B_{(2)} ,
\end{align}
where $\mathbf{W}_{(1)}$ and $\mathbf{W}_{(2)}$ denote the mode-1 and mode-2 unfoldings of the known core tensor determined by the central inter-layer propagation matrix $\mathbf{W}_{M+1}$. These unfoldings lead to the following LS subproblems:
\begin{align}
J_G(\mathbf{G})
&=
\left\|
\mathbf{Z}_{(1)}
-
\mathbf{G}
\mathbf{T}_R
\mathbf{W}_{(1)}
(\bar{\boldsymbol{\Phi}} \otimes \mathbf{Z}_H)^T
\right\|_F^2 ,
\\
J_H(\mathbf{H})
&=
\left\|
\mathbf{Z}_{(2)}
-
\mathbf{H}^T
\mathbf{T}_L
\mathbf{W}_{(2)}
(\bar{\boldsymbol{\Phi}} \otimes \mathbf{Z}_G)^T
\right\|_F^2 .
\end{align}
leading to the following LS updates:
\begin{align}
\widehat{\mathbf{G}}
&=
\mathbf{Z}_{(1)}
\left(
\mathbf{T}_R
\mathbf{W}_{(1)}
(\bar{\boldsymbol{\Phi}} \otimes \widehat{\mathbf{H}}^T\mathbf{T}_L)^T
\right)^+,\label{eq:updateGtucker}
\\
\widehat{\mathbf{H}}^T
&=
\mathbf{Z}_{(2)}
\left(
\mathbf{T}_L
\mathbf{W}_{(2)}
(\bar{\boldsymbol{\Phi}} \otimes \widehat{\mathbf{G}}\mathbf{T}_R)^T
\right)^+. \label{eq:updateHtucker}
\end{align}
It is noted that, when $\bar{\mathbf \Phi}$ is fixed by training and known to the Rx, ALS therefore alternates only between $\mathbf G$ and $\mathbf H$.

\noindent\textit{Remark 3}: The proposed TenSIM-Tucker estimator can recover the physical channel matrices $\mathbf G$ and $\mathbf H$ accurately when the SIM has a relatively small number of meta-atoms and the inter-layer distance is small. In this regime, the LS updates remain numerically well-behaved, so the estimator can recover $\mathbf G$ and $\mathbf H$ up to a single scalar factor. However, as $N$ or $d_{\mathrm{layer}}$ increases (see the Rayleigh--Sommerfeld model in (\ref{eq:rayleigh_sommerfeld})), the relevant inter-layer propagation matrices become increasingly ill conditioned, especially the central matrix $\mathbf W_{M+1}$ that enters the Tucker LS updates through its mode-1, $\mathbf W_{(1)}$, and mode-2, $\mathbf W_{(2)}$, unfoldings. Consequently, when $\mathbf W_{M+1}$ is ill-conditioned, the associated unfoldings $\mathbf W_{(1)}$ and $\mathbf W_{(2)}$ also become ill-conditioned, which in turn deteriorates the conditioning of the LS subproblems in (\ref{eq:updateGtucker}) and (\ref{eq:updateHtucker}). This loss of conditioning amplifies numerical errors in pseudoinverse-based updates, degrading the estimator's stability. Therefore, although the Tucker estimator remains valid and effective for moderately sized SIM cascades with small inter-layer spacing, its sensitivity to the conditioning of the propagation matrices becomes a limiting factor in larger or widely spaced even-layer SIM architectures.

\subsection{The Case of Unknown/Imperfect SIM Configuration}\label{sec:blind-SIM-variant}
In practical implementations, the SIM training matrix may be unavailable or not perfectly known at the Rx due to phase quantization errors, controller latency, RF chain mismatch, calibration drift, mutual coupling, and other hardware impairments. In this case, a robust alternative is to jointly estimate the training-related factor with the channel-side factors, as follows.

For the odd-layer PARAFAC model, re-activate the mode-3 LS subproblem as:
\begin{equation}
J_{\Phi}(\mathbf \Phi)=\left\|\mathbf Z_{(3)}-\mathbf \Phi(\mathbf Z_H \diamond \mathbf Z_G)^T\right\|_F^2,
\label{eq:JPhi_odd}
\end{equation}
whose LS minimizer is given by:
\begin{equation}
\widehat{\mathbf \Phi}=\mathbf Z_{(3)}\Big(((\widehat{\mathbf H}^T\mathbf S_{\textrm{L}})\diamond(\widehat{\mathbf G}\mathbf S_{\textrm{R}}))^T\Big)^+.
\label{eq:LS_Phi_odd}
\end{equation}
Thus, ALS cycles over $\{\mathbf G,\mathbf H,\mathbf \Phi\}$ instead of only $\{\mathbf G,\mathbf H\}$.

For the even-layer Tucker model, the analogous mode-3 LS cost function is defined as follows:
\begin{equation}
J_{\bar{\Phi}}(\bar{\mathbf \Phi})=\left\|\mathbf Z_{(3)}-\bar{\mathbf \Phi}\mathbf W_{(3)}(\mathbf Z_H\otimes \mathbf Z_G)^T\right\|_F^2,
\label{eq:JPhi_even}
\end{equation}
with LS update given by:
\begin{equation}
\widehat{\bar{\mathbf \Phi}}=\mathbf Z_{(3)}\Big({\mathbf W}_{(3)}(\widehat{\mathbf H}^T\mathbf T_{\textrm{L}}\otimes\widehat{\mathbf G}\mathbf T_{\textrm{R}})^T\Big)^+.
\label{eq:LS_Phi_even}
\end{equation}
Therefore, the Tucker-ALS routine alternates over $\{\mathbf G,\mathbf H,\bar{\mathbf \Phi}\}$ when the training matrix is uncertain.

\revblue{In the simulations, the stopping tolerance is applied to the relative reconstruction error $e_i=\|\ten Z-\widehat{\ten Z}^{(i)}\|_F^2/\|\ten Z\|_F^2$, and the iterations stop when the relative decrease of $e_i$ falls below $\epsilon$ or when $I_{\max}$ is reached. The nominal implementation uses algebraic/SVD-based starts when available. Otherwise, random full-column-rank initial factors are used, followed by column normalization to control the scaling ambiguity. For odd-layer SIMs, the ALS routine outputs the effective factors $\widehat{\mathbf Z}_G$ and $\widehat{\mathbf Z}_H$ because the fixed sub-cascades are absorbed into them. For even-layer SIMs, the known Tucker core and calibrated fixed sub-cascades allow the estimates to be interpreted as $\widehat{\mathbf G}$ and $\widehat{\mathbf H}$, equivalently through $\widehat{\mathbf Z}_G=\widehat{\mathbf G}\mathbf T_{\textrm{R}}$ and $\widehat{\mathbf Z}_H=\widehat{\mathbf H}^T\mathbf T_{\textrm{L}}$. In both cases, the final performance metric is computed from the reconstructed cascaded channel, for which the intrinsic scaling ambiguity cancels.}

\begin{algorithm}[t]
\caption{TenSIM Channel Estimation Framework}
\label{alg:unified_sim_als}
\begin{algorithmic}[1]
\Require \algwrap{Received signal $\ten Z$, number of layers $L$,\\ stopping threshold $\epsilon$, maximum iterations $I_{\max}$;\\ SIM training matrix ($\mathbf \Phi$ or $\bar{\mathbf \Phi}$) when it is known,\\ or an initialization of this factor when unknown.}
\Ensure Estimated factors $\widehat{\mathbf G}$ and $\widehat{\mathbf H}$. 
\If{$L=2M+1$ (odd-layer case)}
    \State \algwrap{Build the PARAFAC model in (\ref{eq:ZtensorCP}) with factors $\mathbf Z_G=\mathbf G\mathbf S_{\textrm{R}}$ and $\mathbf Z_H=\mathbf H^T\mathbf S_{\textrm{L}}$.}
    \State \algwrap{Initialize $\widehat{\mathbf G}$ and $\widehat{\mathbf H}$; if the SIM training matrix is unknown, also initialize $\widehat{\mathbf \Phi}$.}
    \For{$i=1,\dots,I_{\max}$}
        \State Update $\widehat{\mathbf G}$ using (\ref{eq:LS_G_odd}).
        \State Update $\widehat{\mathbf H}$ using (\ref{eq:LS_H_odd}).
        \State \algwrap{{\it Optional:} Update $\widehat{\mathbf \Phi}$ using (\ref{eq:LS_Phi_odd}) only if\\the SIM training matrix is unknown.}
        \If{relative reconstruction-error decrease $<\epsilon$}
            \State \textbf{break}
        \EndIf
    \EndFor
\Else \Comment{$L=2M$, even-layer case}
    \State \algwrap{Build the Tucker model in (\ref{eq:ZtensorCP_even}) with factors \\ $\mathbf Z_G=\mathbf G\mathbf T_{\textrm{R}}$ and $\mathbf Z_H=\mathbf H^T\mathbf T_{\textrm{L}}$.}
    \State \algwrap{Initialize $\widehat{\mathbf G}$ and $\widehat{\mathbf H}$; if the SIM training matrix is unknown, also initialize $\widehat{\bar{\mathbf \Phi}}$.}
    \For{$i=1,\dots,I_{\max}$}
        \State Update $\widehat{\mathbf G}$ using (\ref{eq:updateGtucker}).
        \State Update $\widehat{\mathbf H}$ using (\ref{eq:updateHtucker}).
        \State \algwrap{{\it Optional:} Update $\widehat{\bar{\mathbf \Phi}}$ using (\ref{eq:LS_Phi_even}) only if \\ the SIM training matrix is unknown.} 
        \If{relative reconstruction-error decrease $<\epsilon$}
            \State \textbf{break}
        \EndIf
    \EndFor
\EndIf
\end{algorithmic}
\end{algorithm}

\subsubsection*{Convergence Considerations} The TenSIM ALS estimators used in both the PARAFAC (odd-layer) and Tucker (even-layer) cases, summarized in Algorithm~\ref{alg:unified_sim_als}, inherit the standard monotonic non-increase of the LS objective under exact block updates. In the common operational regime where the SIM training factor is known (or treated as fixed), convergence is typically reached in a small number of iterations (often a few to a few tens), with the exact count depending mainly on the signal-to-noise ratio (SNR) and the conditioning of the design matrices. In the blind-SIM variant, where $\mathbf \Phi$ (or $\bar{\mathbf \Phi}$) is jointly estimated with the channel matrices, convergence becomes more sensitive to initialization because of the additional multilinear coupling and scale/permutation ambiguities. Still, practical robustness can be significantly improved through informed initializations (e.g., algebraic/singular-value-decomposition-based starts), safeguarded update rules (e.g., line-search/damping), and mature tensor-optimization toolboxes, such as Tensorlab \cite{Tensorlab2016}. In particular, a warm start for the SIM training factor using the nominal phase codebook is often effective in practice, accelerating convergence and reducing poor local-stationary outcomes.

\subsection{Computational Complexity}
Let $I_{\max}$ denote the number of ALS iterations until convergence. The dominant cost in each iteration comes from forming LS design matrices (Khatri-Rao/Kronecker products), multiplying unfoldings by these matrices, and computing right inverses of the resulting matrix systems with $N$ rows.

For the following odd-layer PARAFAC model, the updates of $\mathbf G$ and $\mathbf H$ use the design matrices:
\begin{equation}\label{eq:A_G_and_A_H}
\begin{aligned}
\mathbf A_G &\triangleq\mathbf S_{\textrm{R}}(\mathbf \Phi\diamond \mathbf Z_H)^T\in\mathbb C^{N\times KM_T},\\
\mathbf A_H &\triangleq\mathbf S_{\textrm{L}}(\mathbf \Phi\diamond \mathbf Z_G)^T\in\mathbb C^{N\times KM_R}.
\end{aligned}
\end{equation}
Hence, the per-iteration complexity is dominated by
\begin{equation}
O\!\left(N^2K(M_T+M_R) + N^3\right),
\end{equation}
where the $N^2K(\cdot)$ terms account for LS matrix products and the $N^3$ term comes from solving/pseudo-inverting the normal equations. Therefore, over $I_{\max}$ iterations, the overall complexity is:
\begin{equation}
C_{\text{odd}}=O\!\left(I_{\max}\left(N^2K(M_T+M_R)+N^3\right)\right).
\end{equation}
If $\mathbf \Phi$ is fixed and known, the third update is skipped, further reducing runtime and memory traffic.

For the even-layer Tucker model, the updates involve
\begin{equation}
\begin{aligned}
\mathbf B_G &\triangleq\mathbf T_{\textrm{R}}\mathbf W_{(1)}(\bar{\mathbf \Phi}\otimes \mathbf Z_H)^T,\\
\mathbf B_H &\triangleq\mathbf T_{\textrm{L}}\mathbf W_{(2)}(\bar{\mathbf \Phi}\otimes \mathbf Z_G)^T,
\end{aligned}
\end{equation}
where the Kronecker factors and the core unfoldings $\mathbf W_{(1)}$ and $\mathbf W_{(2)}$ increase arithmetic complexity and memory costs. A representative per-iteration order is:
\begin{equation}
O\!\left(N^2K(M_T+M_R) + N^3 + N^2\,\mathrm{nnz}(\mathbf W_{(1:2)})\right),
\end{equation}
with $\mathrm{nnz}(\mathbf {W}_{(1:2)})$ denoting the effective number of nonzero core coefficients used in the two channel-side updates. By design, $\mathrm{nnz}(\mathbf W_{(1:2)})=N^2$, so the per-iteration order simplifies to $ O\!\left(N^2K(M_T+M_R) + N^3 + N^4\right)$.
Thus, the overall complexity is:
\begin{equation}
 C_{\text{even}}= O\!\left(I_{\max}\left(N^2K(M_T+M_R)+N^3+N^4\right)\right).
\end{equation}

The latter complexity expressions highlight an important trade-off. The TenSIM-PARAFAC estimator is computationally lighter and scales more favorably with system size, whereas the TenSIM-Tucker estimator incurs higher complexity because of the coupled core tensor updates. In practice, the preferred operating regime depends on the application requirements. For fast tracking and large-scale arrays, using an odd number of layers together with the PARAFAC-based estimator is attractive. When the even-layer structure is required, the Tucker-based estimator justifies its additional cost by providing a more faithful model of the SIM channel. 

\subsection{Impact of the SIM Configuration on Numerical Stability}
While both the odd- and even-layer models originate from the same SIM cascade physics, their numerical behavior differs significantly due to the distinct algebraic structures induced by the training protocol. 
For the odd-layer architecture, the resulting received signal model admits a PARAFAC model in which the effective channel factors $\mathbf{Z}_G$ and $\mathbf{Z}_H$ appear linearly in the tensor decomposition, as shown in expression~(\ref{eq:oddlayer_parafac_model}). In this case, the estimation updates involve Khatri–Rao products between the central-layer SIM training matrix $\mathbf{\Phi}$ and the channel factors. Since this training matrix can be designed to have full column rank and favorable conditioning, the resulting LS problems remain well-behaved even for relatively large SIM sizes. In contrast, the even-layer architecture yields a Tucker model given by expression (\ref{eq:ZtensorCP_even}), whose core tensor is determined by the inter-layer propagation matrix $\mathbf{W}_{M+1}$. 
As a result, the numerical conditioning of the LS problems in the ALS updates becomes directly influenced by the singular value spectrum of $\mathbf{W}_{M+1}$. This matrix tends to become increasingly ill-conditioned as the number of SIM elements or the inter-layer distance increases. 
In this case, the electromagnetic coupling between layers concentrates energy into a smaller set of spatial modes, effectively reducing the numerical rank of the interlayer propagation matrices. Consequently, ALS updates for TenSIM-Tucker become increasingly sensitive to numerical errors. Hence, this estimator performs more reliably when the inter-layer spacing remains small, and the SIM aperture is moderate. 

\revblue{It is important to distinguish algebraic identifiability from stable finite-SNR estimation. The rank conditions derived in Section~\ref{sec:ident} guarantee uniqueness in an ideal, noiseless sense, but they do not prevent performance degradation when the Khatri--Rao/Kronecker design matrices are nearly rank-deficient or poorly conditioned. In such cases, the ALS algorithm may converge slowly. This effect is more pronounced for the Tucker branch because the core unfoldings inherit the conditioning of $\mathbf W_{M+1}$, whereas the PARAFAC branch places the fixed propagation operators inside the effective factors and avoids an explicit nonidentity core in the LS updates.}

\section{Identifiability Study}\label{sec:ident}
The ALS updates in both odd- and even-layer cases are computed through right pseudo-inverses. Therefore, existence and uniqueness (in the LS sense) are governed by the rank of the corresponding design matrices. 


\begin{proposition}[Odd-layer LS identifiability conditions]
If $\operatorname{rank}(\mathbf A_G)=N$ and $\operatorname{rank}(\mathbf A_H)=N$ in~\eqref{eq:A_G_and_A_H}, then the LS updates in (\ref{eq:LS_G_odd}) and (\ref{eq:LS_H_odd}) are uniquely defined (for fixed companion factors). A necessary dimension condition is $KM_T\geq N$ and $KM_R\geq N$.
\end{proposition}

For the even-layer (Tucker) case, recall that the effective factors satisfy $\mathbf Z_G=\mathbf G\mathbf T_{\textrm{R}}$ and $\mathbf Z_H=\mathbf H^T\mathbf T_{\textrm{L}}$. Equivalently, the even-layer regressors may be expressed in terms of $\mathbf Z_G$ and $\mathbf Z_H$, but we write them below in terms of $\mathbf G$ and $\mathbf H$ to match the direct LS updates in (\ref{eq:updateGtucker}) and (\ref{eq:updateHtucker}), as follows:
\begin{eqnarray}
\mathbf B_G &\triangleq& \mathbf T_{\textrm R}\mathbf W_{(1)}(\bar{\mathbf \Phi}\otimes \mathbf H^T\mathbf T_{\textrm L})^T \in \mathbb{C}^{N \times KM_T},\label{eq:80}\\
\mathbf B_H &\triangleq& \mathbf T_{\textrm L}\mathbf W_{(2)}(\bar{\mathbf \Phi}\otimes \mathbf G\mathbf T_{\textrm R})^T \in \mathbb{C}^{N \times KM_R}.\label{eq:81}
\end{eqnarray}

\begin{proposition}[Even-layer LS identifiability conditions]
If $\operatorname{rank}(\mathbf B_G)=N$ and $\operatorname{rank}(\mathbf B_H)=N$, then the LS updates in (\ref{eq:updateGtucker}) and (\ref{eq:updateHtucker}) are uniquely defined (for fixed companion factors). A necessary dimension condition is $KM_T\geq N$ and $KM_R\geq N$.
\end{proposition}

\begin{corollary}[Odd-layer sufficient full-rank conditions]
If $\mathbf S_{\textrm{R}}$ and $\mathbf S_{\textrm{L}}$ are nonsingular and it holds $\operatorname{rank}(\mathbf \Phi\diamond \mathbf Z_H)=N$ and 
$\operatorname{rank}(\mathbf \Phi\diamond \mathbf Z_G)=N$,
then, $\operatorname{rank}(\mathbf A_G)=\operatorname{rank}(\mathbf A_H)=N$.
This observation links the LS identifiability conditions directly to the SIM training design. Specifically, if the training matrix $\mathbf \Phi\in\mathbb{C}^{K\times N}$ is designed to have orthogonal columns and, hence, full column rank with $K\geq N$, while $\mathbf Z_H$ and $\mathbf Z_G$ are also full column rank, then the standard rank property of the Khatri--Rao product \cite{SidiropoulosBro2000} guarantees that the following conditions hold: 
\begin{equation}
\operatorname{rank}(\mathbf \Phi\diamond \mathbf Z_H)=\operatorname{rank}(\mathbf \Phi\diamond \mathbf Z_G)=N.\label{eq:ranksAGAH}
\end{equation}
Therefore, Proposition~1 is satisfied whenever the left cascade factors $\mathbf S_{\textrm{R}}$ and $\mathbf S_{\textrm{L}}$ do not reduce rank. This requirement is physically meaningful because these matrices are built from products of known SIM inter-layer propagation matrices and fixed diagonal SIM coefficient matrices. Under the Rayleigh--Sommerfeld interlayer propagation model, the coupling matrices are generically full rank for non-degenerate SIM geometries and element placements. Hence, if the corresponding diagonal SIM coefficient matrices also have nonzero diagonal entries, then $\mathbf S_{\textrm{R}}$ and $\mathbf S_{\textrm{L}}$ are nonsingular. Under these conditions, the LS updates in \eqref{eq:LS_G_odd} and \eqref{eq:LS_H_odd} are uniquely defined, and the conditions of Proposition~1 guarantee the uniqueness of the corresponding channel estimates.
\end{corollary}

\begin{corollary}[Even-layer sufficient full-rank conditions]
If in \eqref{eq:80} and \eqref{eq:81} holds that: $\operatorname{rank}(\mathbf T_{\textrm R})=\operatorname{rank}(\mathbf T_{\textrm L})=N$, $\operatorname{rank}(\bar{\mathbf \Phi}\otimes \mathbf H^T\mathbf T_{\textrm L})\ge N$, $\operatorname{rank}(\bar{\mathbf \Phi}\otimes \mathbf G\mathbf T_{\textrm R})\ge N$, $\operatorname{rank}(\mathbf W_{(1)})\ge N$, and $\operatorname{rank}(\mathbf W_{(2)})\ge N$,
with no cancellation of the relevant $N$-dimensional subspaces after left multiplication by $\mathbf T_{\textrm R}\mathbf W_{(1)}$ and $\mathbf T_{\textrm L}\mathbf W_{(2)}$, then $\operatorname{rank}(\mathbf B_G)=\operatorname{rank}(\mathbf B_H)=N$.
This condition can be linked directly to the proposed design of the SIM training matrix. If $\bar{\mathbf \Phi}\in\mathbb{C}^{K\times N}$ is designed with orthogonal columns, and therefore has full column rank for $K\ge N$, then the Kronecker rank identity $\operatorname{rank}(\mathbf A\otimes\mathbf B)=\operatorname{rank}(\mathbf A)\operatorname{rank}(\mathbf B)$ implies that:
\begin{eqnarray}
&&\operatorname{rank}(\bar{\mathbf \Phi}\otimes \mathbf H^T\mathbf T_{\textrm L})=\operatorname{rank}(\bar{\mathbf \Phi})\operatorname{rank}(\mathbf H^T\mathbf T_{\textrm L})\ge N,\nonumber\\
&&\operatorname{rank}(\bar{\mathbf \Phi}\otimes \mathbf G\mathbf T_{\textrm R})=\operatorname{rank}(\bar{\mathbf \Phi})\operatorname{rank}(\mathbf G\mathbf T_{\textrm R})\ge N.\nonumber
\end{eqnarray}
Hence, the rank conditions on $\mathbf B_G$ and $\mathbf B_H$ are ensured whenever the cascade factors $\mathbf T_{\textrm R}$ and $\mathbf T_{\textrm L}$ are nonsingular and the unfoldings $\mathbf W_{(1)}$ and $\mathbf W_{(2)}$ preserve these $N$-dimensional subspaces. Since $\mathbf W_{(1)}$ and $\mathbf W_{(2)}$ are determined by the middle SIM coupling core, a sufficient physical condition is that this core has mode-1 and mode-2 unfoldings of rank at least $N$, and that these unfoldings do not annihilate the Kronecker-product subspaces above. By the same Rayleigh--Sommerfeld argument as in the odd-layer case, these unfoldings are generically full rank for non-degenerate SIM geometries and element placements. In addition, if the fixed left and right SIM sub-cascades are built from nonsingular inter-layer propagation matrices and diagonal SIM coefficient matrices with nonzero diagonal entries, then $\mathbf T_{\textrm R}$ and $\mathbf T_{\textrm L}$ are also nonsingular. Under these conditions, the LS updates in (\ref{eq:updateGtucker}) and (\ref{eq:updateHtucker}) are uniquely defined, and the conditions of Proposition~2 guarantee the uniqueness of the corresponding channel estimates.
\end{corollary}

In summary, under the training design and cascade conditions discussed in Corollaries 1 and 2, it holds that:
\begin{equation}
\operatorname{rank}(\mathbf A_G)=\operatorname{rank}(\mathbf A_H)=\operatorname{rank}(\mathbf B_G)=\operatorname{rank}(\mathbf B_H)=N.\nonumber
\end{equation}
Therefore, the inequalities stated in Propositions~1 and~2 provide sufficient conditions for identifiability, which are useful for system design, since they explicitly involve key parameters such as the numbers of transmit ($M_T$) and receive ($M_R$) antennas, the number of SIM meta-atoms ($N$), and the number of training blocks ($K$).

\textit{Remark 4}: When $K<N$ (a practically relevant short-training regime), neither $\mathbf \Phi$ nor $\bar{\mathbf \Phi}$ can be full column rank and, therefore, the sufficient conditions above no longer apply directly. For the odd-layer model in \eqref{eq:JG} and \eqref{eq:JH}, one can still exploit the Khatri--Rao structure: if the companion factors have full column rank, then their Khatri--Rao product remains full column rank. Hence, even when $\mathbf \Phi$ is rank-deficient because $K<N$, identifiability can still be studied through alternative sufficient conditions involving full column rank of $\mathbf Z_H$ and $\mathbf Z_G$, together with nonsingular cascade factors, so that the design matrices in \eqref{eq:JG} and \eqref{eq:JH} remain full row rank. For the even-layer model, an analogous interpretation follows from the Kronecker-product structure. Although $\bar{\mathbf \Phi}$ is rank-deficient when $K<N$, the matrices $\bar{\mathbf \Phi}\otimes (\mathbf H^T\mathbf T_{\textrm L})$ and $\bar{\mathbf \Phi}\otimes (\mathbf G\mathbf T_{\textrm R})$ may still contain relevant $N$-dimensional subspaces, provided that the effective factors $\mathbf H^T\mathbf T_{\textrm L}$ and $\mathbf G\mathbf T_{\textrm R}$ are sufficiently rich and the products $\mathbf T_{\textrm R}\mathbf W_{(1)}$ and $\mathbf T_{\textrm L}\mathbf W_{(2)}$ preserve those subspaces. Therefore, in the short-training regime, identifiability should be characterized by structure-aware rank conditions on the physical channel factors, together with the fixed SIM sub-cascades and the middle coupling core, rather than by full-column-rank training alone.

\subsection{Scaling Ambiguities}
The ambiguity structure is fundamentally different in the odd- and even-layer SIM models. In this section, we characterize the remaining indeterminacies of the tensor factorizations induced by the proposed training protocol.

\paragraph*{Odd-Layer SIM (PARAFAC model)}
For the odd-layer architecture, the received tensor admits the PARAFAC representation:
\begin{equation}
\ten Z = \ten I_{3,N}
\times_1 \mathbf Z_G
\times_2 \mathbf Z_H
\times_3 \mathbf \Phi ,
\end{equation}
where $\mathbf Z_G=\mathbf G\mathbf S_{\textrm R}$ and
$\mathbf Z_H=\mathbf H^T\mathbf S_{\textrm L}$.

\begin{proposition}
Consider the PARAFAC model above, where the third factor matrix $\mathbf \Phi$ is known and fixed by the SIM training design.
Then, the factor matrices $\mathbf Z_G$ and $\mathbf Z_H$ are identifiable only up to a diagonal scaling transformation.
Specifically, for any nonsingular diagonal matrix $\mathbf\Lambda\in\mathbb C^{N\times N}$, the following equations hold:
\begin{equation}
\widehat{\mathbf Z}_G=\mathbf Z_G\mathbf\Lambda,
\qquad
\widehat{\mathbf Z}_H=\mathbf Z_H\mathbf\Lambda^{-1}.\label{eq:scaling_amb}
\end{equation}
\end{proposition}


The implication of this ambiguity is that the effective SIM-channel
factors $\mathbf Z_G$ and $\mathbf Z_H$ can only be determined up to
diagonal scaling.
Since $\mathbf Z_G=\mathbf G\mathbf S_{\textrm R}$ and
$\mathbf Z_H=\mathbf H^T\mathbf S_{\textrm L}$, this ambiguity cannot, in general, be transferred directly to the
physical channels $\mathbf G$ and $\mathbf H$. This happens because the diagonal matrix $\mathbf\Lambda$ does not commute with the SIM cascade matrices $\mathbf S_{\textrm R}$ and $\mathbf S_{\textrm L}$.

\vspace{0.3cm}

\paragraph*{Even-Layer SIM (Tucker model)}
For the even-layer architecture, the received tensor follows the
Tucker model in~\eqref{eq:ZtensorCP_even}:
\begin{equation}
\ten Z
=
\ten W_{M+1}
\times_1 \mathbf Z_G
\times_2 \mathbf Z_H
\times_3 \bar{\mathbf \Phi},\label{eq:Tuckereven}
\end{equation}
where the core tensor $\ten W_{M+1}$ and the training matrix
$\bar{\mathbf \Phi}$ are known.
The frontal slices of $\ten Z$ can therefore be written as:
\begin{equation}
\mathbf Z_k=\mathbf Z_G\,\mathbf Q_k\,\mathbf Z_H^T,
\qquad k=1,\ldots,K,
\end{equation}
where $\mathbf Q_k\triangleq \ten W_{M+1}\times_3 \bar{\boldsymbol{\phi}}_k^T \in \mathbb C^{N\times N}$, with $\bar{\boldsymbol{\phi}}_k$ denoting the $k$-th row of $\bar{\mathbf \Phi}$. Thus, the matrices $\mathbf Q_k$'s are fully known and represent the contraction of the known core tensor with the training applied at the two central layers.

\begin{proposition}
Let $\{\mathbf Q_k\}_{k=1}^K$ be the known matrices defined above.
Two pairs of matrices $(\mathbf Z_G,\mathbf Z_H)$ and
$(\widetilde{\mathbf Z}_G,\widetilde{\mathbf Z}_H)$ produce the
same tensor $\ten Z$ if and only if there exist invertible
matrices $\mathbf T,\mathbf S\in\mathbb C^{N\times N}$ such that: $\widetilde{\mathbf Z}_G=\mathbf Z_G\mathbf T$, $\widetilde{\mathbf Z}_H=\mathbf Z_H\mathbf S$,  
and $\mathbf T\,\mathbf Q_k\,\mathbf S^T=\mathbf Q_k$, $\forall k$.
\end{proposition}

\begin{proof}If the two pairs of matrices generate the same tensor,
their frontal slices must coincide, i.e., 
$\widetilde{\mathbf Z}_G\mathbf Q_k\widetilde{\mathbf Z}_H^T
=
\mathbf Z_G\mathbf Q_k\mathbf Z_H^T$.
Since $\mathbf Z_G$ and $\mathbf Z_H$ are full-column-rank under
the identifiability conditions assumed in the estimation stage,
there exist invertible matrices $\mathbf T$ and $\mathbf S$ such that $\widetilde{\mathbf Z}_G=\mathbf Z_G\mathbf T$ and $\widetilde{\mathbf Z}_H=\mathbf Z_H\mathbf S$.
Substituting into the equality above yields
\begin{equation}
\mathbf Z_G\mathbf T\mathbf Q_k\mathbf S^T\mathbf Z_H^T
=
\mathbf Z_G\mathbf Q_k\mathbf Z_H^T ,
\end{equation}
which implies $\mathbf T\mathbf Q_k\mathbf S^T=\mathbf Q_k $, $\forall k$.
\end{proof}

A relevant case arises when $\mathbf S^T=\mathbf{T}^{-1}$, which preserves the bilinear structure of each slice through a single change of basis acting consistently on the latent dimension shared by $\mathbf Z_G$ and $\mathbf Z_H$. 
In this case, the invariance condition becomes the commutation relation $\mathbf T\mathbf Q_k=\mathbf Q_k\mathbf T,
\qquad \forall k$.
Using $\mathrm{vec}(\mathbf A\mathbf X\mathbf B)
(\mathbf B^T\otimes\mathbf A)\mathrm{vec}(\mathbf X)$, this condition can be written as follows:
\begin{equation}
(\mathbf I\otimes\mathbf Q_k-\mathbf Q_k^T\otimes\mathbf I)
\mathrm{vec}(\mathbf T)=\mathbf 0.
\end{equation}
Stacking the equations $\forall k$ yields:
\begin{equation}
\mathbf C_Q\,\mathrm{vec}(\mathbf T)=\mathbf 0,
\end{equation}
where
\begin{equation}
\mathbf C_Q=
\begin{bmatrix}
\mathbf I\otimes\mathbf Q_1-\mathbf Q_1^T\otimes\mathbf I\\
\mathbf I\otimes\mathbf Q_2-\mathbf Q_2^T\otimes\mathbf I\\
\vdots\\
\mathbf I\otimes\mathbf Q_K-\mathbf Q_K^T\otimes\mathbf I
\end{bmatrix}.
\end{equation}
Therefore, the admissible transformations are precisely those whose vectorized form lies in the nullspace of $\mathbf C_Q$.

\begin{corollary}
If $\dim\ker(\mathbf C_Q)=1$, the only admissible transformation
is $\mathbf T=\alpha\mathbf I$, implying that the $\mathbf Z_G$ and $\mathbf Z_H$ in the even-layer Tucker model (\ref{eq:ZtensorCP_even}) are identifiable up to a single global scalar factor.
\end{corollary}

The nullity condition can be linked explicitly to the physical model. For a fixed middle propagation matrix $\mathbf W_{M+1}$, generated here by the Rayleigh--Sommerfeld model, each slice $\mathbf {Q}_{k} $ is obtained by contracting the known core tensor $\ten W_{M+1}$ with the training vector associated with the two central layers. Hence, the entries of $\mathbf Q_k$ depend polynomially and, in fact, linearly on the entries of $\mathbf W_{M+1}$ and on the components of the phase vector $\bar{\boldsymbol\phi}_k$. Consequently, the condition $\dim\ker(\mathbf C_Q)>1$ is equivalent to the existence of a non-scalar matrix $\mathbf T$ satisfying $\mathbf T\mathbf Q_k=\mathbf Q_k\mathbf T$, $\forall k$, which imposes polynomial constraints on the entries of $\mathbf W_{M+1}$ and on the training phases.

\begin{corollary}
Assume that the Rayleigh--Sommerfeld matrix $\mathbf W_{M+1}$ is fixed and nondegenerate.
If the phase vectors of the two central layers are drawn independently from an absolutely continuous distribution on the unit circle, then, with probability one, the resulting matrices $\{\mathbf Q_k\}_{k=1}^K$ satisfy
$\dim\ker(\mathbf C_Q)=1$.
Equivalently, the residual ambiguity of the even-layer Tucker model reduces almost surely to a single global scalar factor under generic random phase training.
\end{corollary}
\begin{proof}
For fixed $\mathbf W_{M+1}$, the entries of $\mathbf C_Q$ are polynomial functions of the random phase variables. The event $\dim\ker(\mathbf C_Q)>1$ is equivalent to the vanishing of all minors of $\mathbf C_Q$ of the maximal rank compatible with one-dimensional nullity and is, therefore, an algebraic event in the training phases. Under the nondegeneracy assumption on $\mathbf W_{M+1}$, these minors are not all identically zero, so the exceptional set where $\dim\ker(\mathbf C_Q)>1$ is a proper algebraic subset of the phase-design space. Since the phase vectors are drawn from an absolutely continuous distribution, this subset has probability zero. Thus, $\dim\ker(\mathbf C_Q)=1$ almost surely.
\end{proof}

\begin{table*}[!t]
\centering
\small
\caption{The Proposed TenSIM Channel Estimation Methods: PARAFAC (odd-layer case) vs. Tucker (even-layer case).}
\label{tab:odd_even_comparison}
\vspace{1ex}
\begin{tabular}{>{\raggedright\arraybackslash}p{0.20\textwidth}>{\raggedright\arraybackslash}p{0.33\textwidth}>{\raggedright\arraybackslash}p{0.33\textwidth}}
\hline
\textbf{Aspect} & \textbf{TenSIM-PARAFAC} ($L=2M+1$) & \textbf{TenSIM-Tucker} ($L=2M$) \\
\hline
Tensor model & PARAFAC with identity core & Tucker with nontrivial core $\ten W_{M+1}$ \\
Training-varying layers & Single middle layer ($M+1$) & Two middle layers ($M$ and $M+1$) \\
\revblue{Main estimated quantities} & \revblue{Effective factors $\mathbf Z_G,\mathbf Z_H$ and cascaded-channel reconstructions} & \revblue{Physical/effective channel factors $\mathbf G,\mathbf H$ or $\mathbf Z_G,\mathbf Z_H$ and cascaded-channel reconstructions} \\
Identifiability behavior & Essentially unique (diagonal scaling factor) & Essentially unique (global scalar factor)\\
Computational profile & Lower per-iteration cost and memory & Higher per-iteration cost  \\
\hline
\end{tabular}
\end{table*}

\begin{table*}[!t]
\centering
\setlength{\tabcolsep}{3pt}
\renewcommand{\arraystretch}{1.12}
\small
\caption{\revblue{Comparison between the proposed TenSIM estimators and LS reference methods.}}
\label{tab:method_baseline_comparison}
\vspace{1ex}
\begin{tabular}{>{\raggedright\arraybackslash}p{0.16\textwidth}>{\raggedright\arraybackslash}p{0.28\textwidth}>{\raggedright\arraybackslash}p{0.18\textwidth}>{\raggedright\arraybackslash}p{0.18\textwidth}>{\raggedright\arraybackslash}p{0.10\textwidth}}
\hline
\revblue{\textbf{Method}} & \revblue{\textbf{Required CSI/side information}} & \revblue{\textbf{SIM structure exploited?}} & \revblue{\textbf{Estimated variables}} & \revblue{\textbf{Works for $K<N^2$?}} \\
\hline
\revblue{TenSIM-PARAFAC} & \revblue{Known pilots, central-layer training factor, calibrated fixed sub-cascades/coupling} & \revblue{Yes} & \revblue{$\mathbf Z_G,\mathbf Z_H$ and cascaded channel} & \revblue{Yes} \\
\revblue{TenSIM-Tucker} & \revblue{Known pilots, two-layer training factor, known Tucker core/fixed sub-cascades} & \revblue{Yes} & \revblue{$\mathbf G,\mathbf H$ or $\mathbf Z_G,\mathbf Z_H$ and cascaded channel} & \revblue{Yes, if rank conditions hold} \\
\revblue{Baseline LS~\cite{CEWCL2024}} & \revblue{Known Tx-SIM channel/subchannel information} & \revblue{Partially} & \revblue{Remaining channel factor or cascaded response} & \revblue{Only with idealized CSI} \\
\revblue{Aggregate LS} & \revblue{Known aggregate SIM training matrix} & \revblue{No} & \revblue{Composite matrix $\mathbf T$} & \revblue{No} \\
\hline
\end{tabular}
\end{table*}

Note that from a system perspective, the scaling ambiguities are not relevant, since the end-to-end SIM channel depends on the matrix product $\mathbf{G}\mathbf{S}\mathbf{H}$, where these scaling factors cancel out. 
The main differences between the odd- and even-layer estimation methods are summarized in Table~\ref{tab:odd_even_comparison}. In short, the odd-layer formulation yields a simpler PARAFAC model with lower complexity, whereas the even-layer formulation yields a Tucker model with richer structure but higher computational cost. \revblue{Finally, Table~\ref{tab:method_baseline_comparison} compares the proposed TenSIM estimators with competing LS baselines, considering assumptions, estimated variables, training requirements, and applicability regimes.}

\section{Numerical Results and Discussion}
This section presents simulation results assessing the performance of the proposed SIM channel estimation methods. The reported figures illustrate the SNR-dependent normalized mean squared error (NMSE) behavior for both odd- and even-layer SIM configurations, the impact of the number of pilot blocks used for channel estimation, the ALS convergence and complexity trends, and the sensitivity to the inter-layer spacing parameter.

\subsection{Simulation Setup}
We have considered a SIM-assisted MIMO communication system following the signal model presented in Section~III. Since the numerical study evaluates several operating regimes, the antenna dimensions $(M_T,M_R)$, the number of SIM elements $N=N_xN_y$, the number of pilot blocks $K$, and the inter-layer spacing $d_{\mathrm{layer}}$ are varied across the experiments. Unless otherwise stated, the SIM was square, i.e., $N_x=N_y$, and the pilot matrix was chosen to be orthogonal, i.e., $\mathbf X\mathbf X^H=\mathbf I_{M_T}$, so that the matched-filtered observations were formed according to \eqref{eq:Zk}. The Tx-to-SIM, $\mathbf H$, and SIM-to-Rx, $\mathbf G$, channels were generated as independent and identically distributed Rayleigh fading matrices and normalized to control the effective SNR. The inter-layer coupling matrices $\{\mathbf W_\ell\}_{\ell=2}^L$ were assumed known and generated from the Rayleigh-Sommerfeld model in~\eqref{eq:rayleigh_sommerfeld}. 
The additive noise follows a circularly symmetric complex Gaussian distribution, and the average SNR was swept over the range specified in each SNR-dependent experiment. For the proposed reduced-complexity SIM training protocol, only the size of the central layer was varied in the odd-layer case ($L=2M+1$), while the sizes of the two central layers were varied in the even-layer case ($L=2M$). Unless otherwise stated, the SIM training matrices were designed according to the procedure discussed in Section~\ref{sec:training_design}. The channel estimation results in all figures were averaged over at least $100$ Monte Carlo realizations. For clarity and reproducibility, the system parameters that remain fixed in each experiment are indicated at the top of the corresponding figure.

\subsection{Performance Evaluation}
To evaluate the performance of the proposed TenSIM estimators, we have computed the average NMSE of the reconstructed cascaded SIM channel across a set of SIM configurations. For the odd-layer SIM architecture, an estimate of the effective end-to-end channel induced by a metasurface configuration $\boldsymbol\phi^{(r)}$ applied to the central layer is defined as follows:
\begin{equation}\label{eq:effective_odd}
\widehat{\mathbf C}_{\mathrm{eff}}^{(r)}
\triangleq
\widehat{\mathbf Z}_G
\mathrm{Diag}(\boldsymbol\phi^{(r)})
\widehat{\mathbf Z}_H^T .
\end{equation}
Note that the actual cascaded channel $\mathbf C_{\mathrm{eff}}^{(r)}$ is invariant to the scaling ambiguities since they cancel out in (\ref{eq:effective_odd}) according to our previous discussion in Section \ref{sec:ident} (see eq.~(\ref{eq:scaling_amb})).
For the even-layer SIM architecture, the estimate of the effective channel induced by a metasurface configuration $\mathbf Q^{(r)}$ applied to the two middle layers can be written as follows:
\begin{equation}\label{eq:effective_even}
\widehat{\mathbf C}_{\mathrm{eff}}^{(r)}
=
\widehat{\mathbf Z}_G
\mathbf Q^{(r)}
\widehat{\mathbf Z}_H^T,
\end{equation}
where $\mathbf Q^{(r)}\in\mathbb{C}^{N\times N}$ denotes the effective configuration induced by the two trained central layers and corresponds to the combined phase vector $\bar{\boldsymbol\phi}^{(r)}=\boldsymbol\phi_M^{(r)}\otimes\boldsymbol\phi_{M+1}^{(r)}$ and the central propagation matrix $\mathbf W_{M+1}$, defined as $\mathrm{vec}(\mathbf Q^{(r)})=
\mathrm{Diag}(\bar{\boldsymbol\phi}^{(r)})\,\mathrm{vec}(\mathbf W_{M+1})$.
Similarly to the odd-layer case, the scaling ambiguity affecting $\mathbf Z_G$ and $\mathbf Z_H$ cancels out in the cascaded channel reconstruction.

\revblue{The reported NMSE values should be interpreted with respect to the estimated object shown in each figure. For the odd-layer PARAFAC branch, the primary channel-related estimates are the effective factors $\mathbf Z_G$ and $\mathbf Z_H$, and the main metric is the reconstructed cascaded channel NMSE in \eqref{eq:effective_odd}. For the even-layer Tucker branch, the figures report the NMSE of the physical channel matrices $\mathbf G$ and $\mathbf H$ when these factors are estimated directly, or the cascaded-channel NMSE in \eqref{eq:effective_even}. Note that in all cascaded-channel evaluations, the intrinsic scaling ambiguity has been automatically removed by the product structure.}

The NMSE performance metric is defined as follows:
\begin{equation}
\mathrm{NMSE}
\triangleq
\frac{1}{R}
\sum_{r=1}^{R}
\frac{
\left\|
\widehat{\mathbf C}_{\mathrm{eff}}^{(r)}
-
\mathbf C_{\mathrm{eff}}^{(r)}
\right\|_F^2
}{
\left\|
\mathbf C_{\mathrm{eff}}^{(r)}
\right\|_F^2
},
\end{equation}
where $R=100$ denotes the number of randomly generated SIM phase realizations used. This metric directly measures the accuracy with which the estimated channels reproduce the end-to-end SIM-induced channel response, while remaining invariant to the intrinsic scaling ambiguities.

\begin{figure}[t]
\centering
\includegraphics[width=\columnwidth]{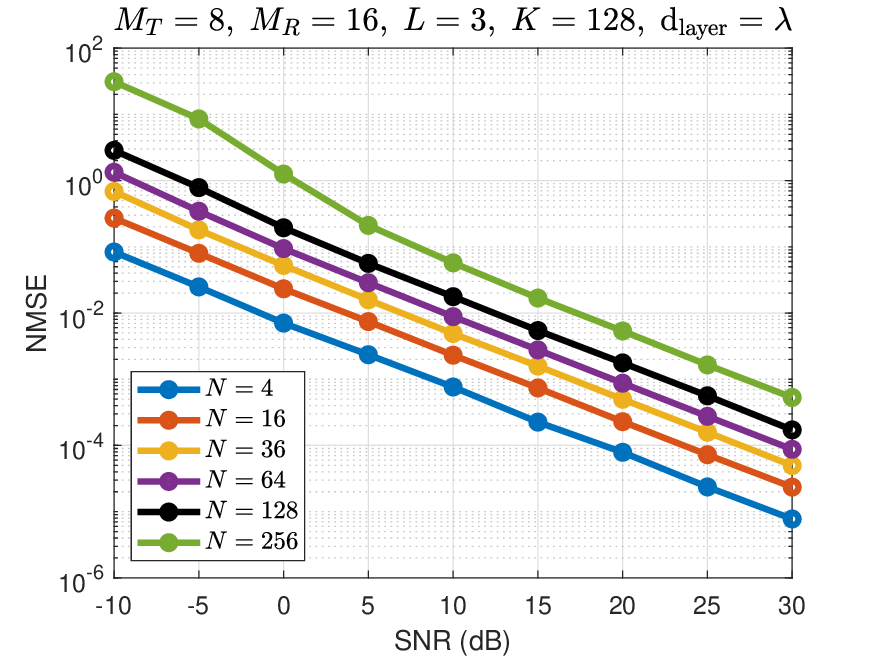}
\caption{Average cascaded-channel NMSE of the proposed TenSIM-PARAFAC versus SNR for the odd-layer SIM setting, illustrating the effect of the number of SIM elements, $N$, on the accuracy of effective channel reconstruction.}
\label{fig:nmse-snr-odd}
\end{figure}
Figure \ref{fig:nmse-snr-odd} demonstrates the effective-channel NMSE of the proposed TenSIM-PARAFAC estimator as a function of the SNR for different SIM sizes ($N=\{4,16,36,64,18,256\}$). As expected, the estimation accuracy improves monotonically with increasing SNR for all configurations. However, the performance gradually degrades as the SIM size $N$ increases. This behavior can be attributed to the increased dimensionality of the cascaded channel as the number of SIM elements grows, which increases the number of parameters involved in the PARAFAC factorization of the cascade. Consequently, larger SIM apertures require higher SNR or more training observations to achieve the same estimation accuracy. Nevertheless, the curves remain approximately parallel and decrease linearly over the tested SNR range, while the SIM dimension primarily affects the overall difficulty of estimation.

\begin{figure}[t]
\centering
\includegraphics[width=\columnwidth]{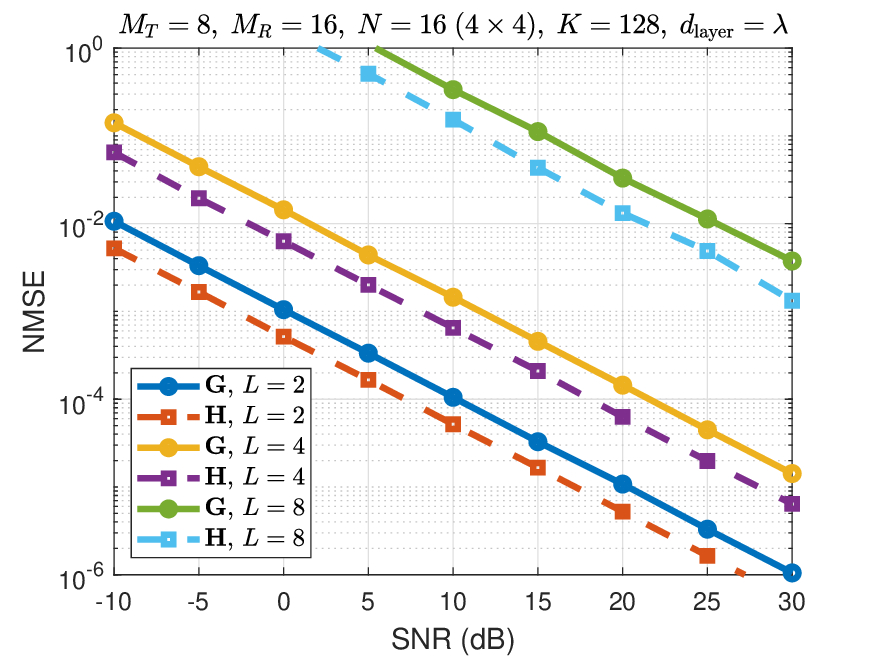}
\caption{Average channel-estimation NMSE of the proposed TenSIM-Tucker estimator versus SNR for the even-layer SIM setting, comparing different SIM depths $L$ under the considered training design.}
\label{fig:nmse-snr-even}
\end{figure}
Figure \ref{fig:nmse-snr-even} shows the NMSE of the estimated channel matrices $\mathbf{G}$ and $\mathbf{H}$ as a function of the SNR for different SIM depths considering the proposed TenSIM-Tucker approach in the even-layer configuration ($L=\{2,4,8\}$). As expected, the estimation accuracy improves with increasing SNR for all configurations. However, the NMSE systematically increases as the number of SIM layers $L$ grows. In particular, the case $L=2$ achieves the most accurate channel estimates, whereas deeper configurations ($L=4$ and $8$) exhibit progressively larger estimation errors. This behavior can be explained by the increased coupling introduced by deeper SIM cascades. As $L$ increases, the effective propagation between the Tx and Rx involves more inter-layer transformations, which are captured in the Tucker model by the matrices $\mathbf{T}_R$ and $\mathbf{T}_L$. These transformations tend to become more ill-conditioned as the cascade depth grows, amplifying noise and numerical errors in the LS updates of the ALS procedure. Consequently, although deeper SIM structures increase the modeling flexibility of the propagation channel~\cite{StylianopoulosOTA2026,SIM_MINN2026}, they also reduce the effective identifiability of the individual channel factors, leading to a gradual degradation in the estimation accuracy of $\mathbf{G}$ and $\mathbf{H}$. Nevertheless, an important advantage of the Tucker formulation in the even-layer case is that it enables the direct recovery of the individual channel matrices $\mathbf{G}$ and $\mathbf{H}$ up to a single scalar ambiguity, a property that is not available in the TenSIM-PARAFAC estimator.

\begin{figure}[t]
\centering
\includegraphics[width=\columnwidth]{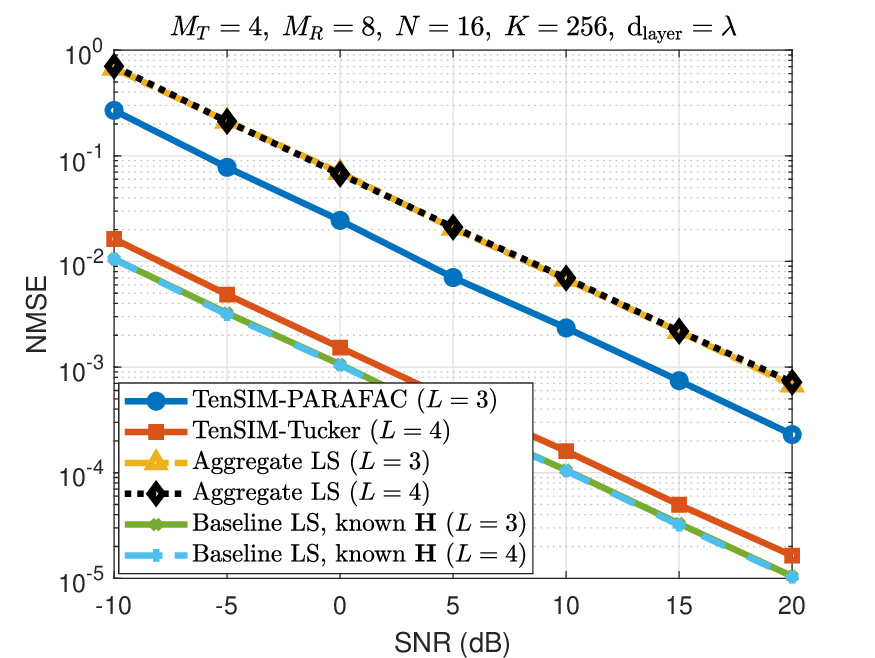}
\caption{Average cascaded-channel NMSE versus SNR for the proposed TenSIM-PARAFAC and TenSIM-Tucker estimators, compared with baseline LS and aggregate LS reference methods.}
\label{fig:nmse-snr-baselines}
\end{figure}
Figure~\ref{fig:nmse-snr-baselines} compares the NMSE performance of the two proposed TenSIM estimators against unstructured LS reference methods as a function of SNR. Two baseline methods are considered. The first (termed ``Baseline LS'') is the LS method proposed in \cite{CEWCL2024}, which assumes that the channel matrix $\mathbf{H}$ linking the Tx to the first SIM layer is known. Note that this is a simplifying assumption valid only when the SIM is placed very close to the Tx. The second reference method (termed ``Aggregate LS'') is the solution given by (\ref{eq:aggLS}) in Section \ref{subsec:3b}, which estimates the aggregate (composite) channel $\mathbf{T}$. 

The results in the figure show that the TenSIM-PARAFAC and TenSIM-Tucker estimators achieve lower NMSE than the aggregate LS baseline. Indeed, the aggregate LS estimator treats $\mathbf{T}=\mathbf H^T\!\otimes\!\mathbf G$ as a fully unstructured MIMO channel matrix, thereby ignoring both the Kronecker structure of the composite channel and the SIM-induced tensor structure of the training matrix. Consequently, its performance remains limited. In contrast, the TenSIM estimators exploit the resulting PARAFAC/Tucker tensor structure of the SIM cascade, leading to more accurate channel reconstruction.

An additional observation is that the aggregate LS baseline is viable in this comparison only because the training length satisfies the full column-rank requirement $K\geq N^2$ for the SIM training matrix $\mathbf S\in\mathbb C^{K\times N^2}$ in (\ref{eq:aggLS}). This requirement is already demanding for moderate SIM sizes and becomes prohibitive as the aperture grows. Therefore, Fig.~\ref{fig:nmse-snr-baselines} should be interpreted as a favorable reference case for the aggregate LS method rather than as a practically scalable solution. We also note that the baseline LS curve with the lowest NMSE relies on the assumed knowledge of the Tx--SIM channel matrix $\mathbf H$, which is idealized, especially when the SIM cascade is placed between the Tx and Rx and the individual subchannels are not directly observable. 

\begin{figure}[t]
\centering
\includegraphics[width=\columnwidth]{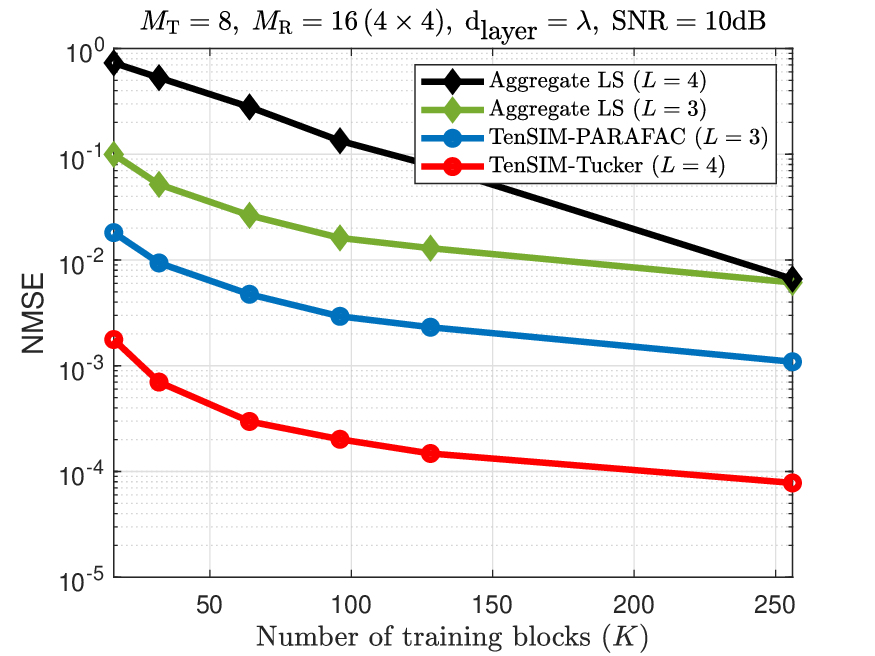}
\caption{Average cascaded-channel NMSE versus the number of training blocks $K$, comparing both TenSIM estimators with the aggregate LS benchmark.}
\label{fig:nmse-snr-baselines2}
\end{figure}

Figure~\ref{fig:nmse-snr-baselines2} further compares TenSIM with the aggregate LS baseline and illustrates the effect of the number of pilot training blocks $K$. It is shown that for small and moderate block sizes, the aggregate LS estimator performs poorly in terms of NMSE. In particular, when $K<N^2$, the resulting estimate is highly sensitive to noise and to the conditioning of the SIM training matrix. In contrast, both proposed TenSIM PARAFAC- and Tucker-based estimators achieve competitive or superior accuracy with a much smaller effective training burden. This is because our estimators leverage the tensor structure induced by the SIM cascade, rather than estimating the composite channel as a fully unstructured matrix.
\begin{figure}[t]
\centering
\includegraphics[width=\columnwidth]{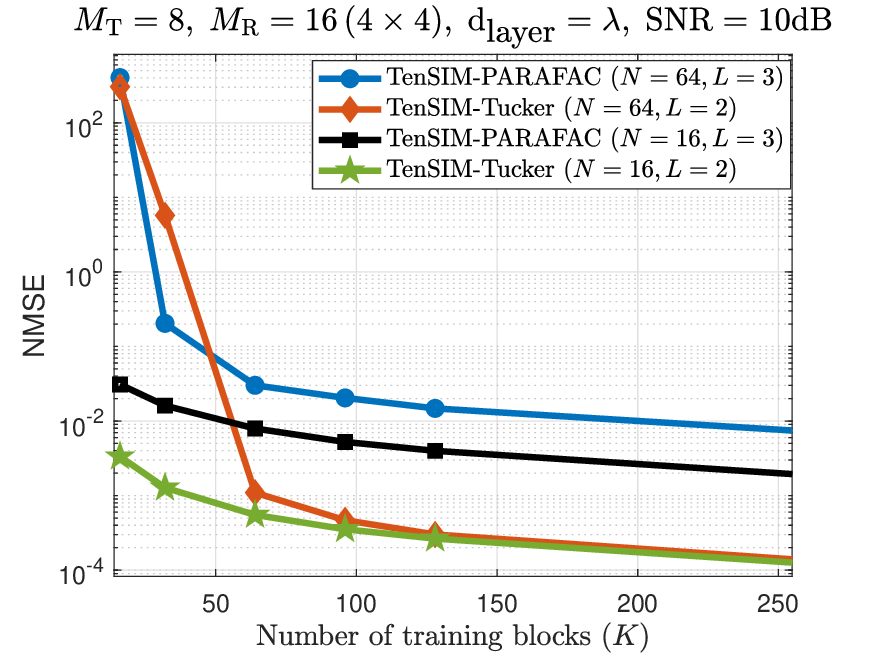}
\caption{Average cascaded-channel NMSE versus the number of pilot blocks $K$ for the proposed TenSIM-PARAFAC and TenSIM-Tucker estimators, showing the impact of training diversity and SIM size.}
\label{fig:nmse-k}
\end{figure}
Further investigation into the role of the training length $K$ in the performance of both proposed estimators is shown in Fig.~\ref{fig:nmse-k}. It is shown therein that for small values of $K$, both estimators exhibit relatively large errors due to the limited pilot resources available to estimate the channel factors relative to the SIM size. However, the NMSE decreases rapidly as $K$ increases. A notable observation is that the TenSIM-Tucker estimator achieves a significantly faster reduction in NMSE once $K$ becomes sufficiently large. In particular, for moderate-to-large values of $K$, the Tucker-based estimator outperforms the PARAFAC-based one. This behavior can be attributed to the richer structural model captured by the Tucker formulation, which accounts for the coupling induced by the SIM layers and enables a more accurate reconstruction of the effective cascaded channel. The influence of the SIM size is also evident from the results. For larger metasurface dimensions (e.g., $N=64$), both estimators require more training blocks to reach the same accuracy, reflecting the increased number of unknown parameters in the channel model. Nevertheless, the overall trends remain consistent: increasing $K$ improves estimation accuracy for both estimators, while the Tucker-based approach benefits more strongly from additional training data because it can exploit the multidimensional tensor structure of the SIM-incorporating channel. Interestingly, this behavior can also be interpreted in terms of the effective number of degrees of freedom of the two tensor models. The PARAFAC formulation used in the odd-layer case imposes a stronger separable structure on the cascaded channel, which reduces the number of free parameters to be estimated. Consequently, reliable estimates can be obtained with a relatively small number of training blocks. In contrast, the Tucker formulation, adopted for the even-layer case, allows a richer representation of the SIM-induced coupling through the core tensor. While this additional modeling flexibility improves reconstruction accuracy when sufficient training data are available, it also increases the number of parameters to be identified. As a result, the Tucker estimator requires more training blocks to fully exploit its structural advantages. Once this regime is reached, however, the Tucker model provides a more accurate representation of the cascaded channel, leading to the superior performance observed for large values of $K$.

\begin{figure}[t]
\centering
\includegraphics[width=\columnwidth]{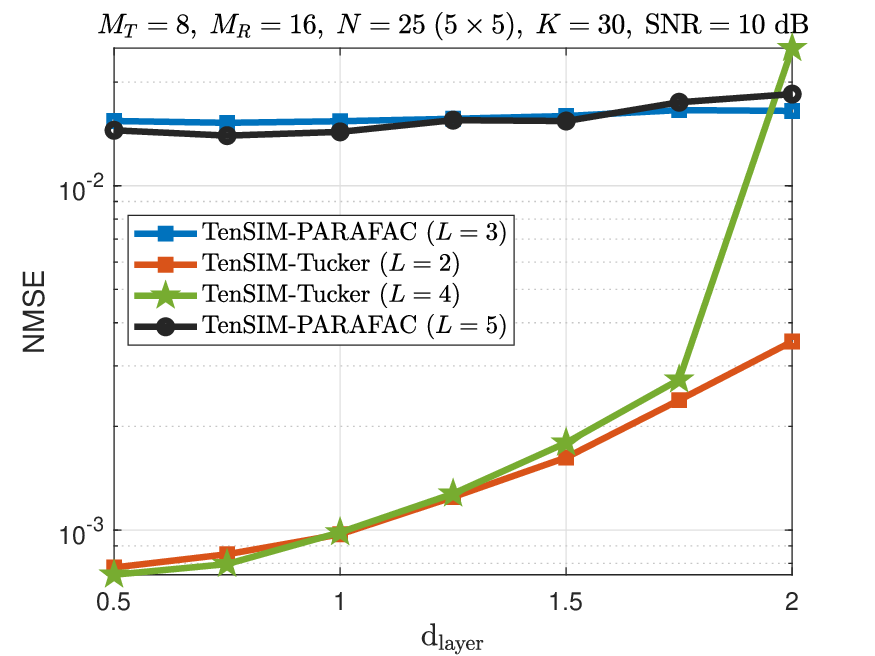}
\caption{Average cascaded-channel NMSE versus the inter-layer spacing $d_{\mathrm{layer}}$, highlighting the sensitivity of the proposed TenSIM-PARAFAC and TenSIM-Tucker estimators to the physical coupling between adjacent SIM layers.}
\label{fig:nmse-dlayer}
\end{figure}
Figure \ref{fig:nmse-dlayer} shows the effective cascaded-channel NMSE as a function of the inter-layer spacing $d_{\mathrm{layer}}$ for different numbers of SIM layers. Two main observations can be made. First, the performance of the TenSIM-PARAFAC estimator remains essentially insensitive to the inter-layer distance over the considered range. The NMSE curves for $L=3$ and $5$ exhibit only minor variations as $d_{\mathrm{layer}}$ increases, indicating that the PARAFAC-based model preserves a stable identifiability structure even as the coupling between adjacent metasurface layers weakens. In contrast, the TenSIM-Tucker estimator exhibits a degradation as the inter-layer spacing increases. For moderate values of $d_{\mathrm{layer}}$, the Tucker-based estimator achieves the best performance among the methods considered. However, as layer separation increases, estimation accuracy gradually deteriorates, with a pronounced performance loss at larger spacing values. This behavior can be explained by the fact that the Tucker formulation explicitly relies on the coupling structure induced by the SIM layers. As the distance between layers increases, the interaction between adjacent SIM elements weakens, thereby reducing the strength of the structured coupling captured by the Tucker core. As a result, the associated LS subproblems used to estimate the individual channels degrade the accuracy of channel reconstruction. Overall, the results in Fig.~\ref{fig:nmse-dlayer} shed light on the trade-off between the two TenSIM branches. While the Tucker-based estimator can exploit the strong structural coupling present in tightly stacked SIM configurations to achieve superior accuracy, its performance becomes more sensitive to the physical separation between layers. On the other hand, the PARAFAC-based estimator appears more robust to variations in inter-layer spacing, although it generally yields slightly higher NMSEs when layer coupling is strong.

\begin{figure}[t]
\centering
\includegraphics[width=\columnwidth]{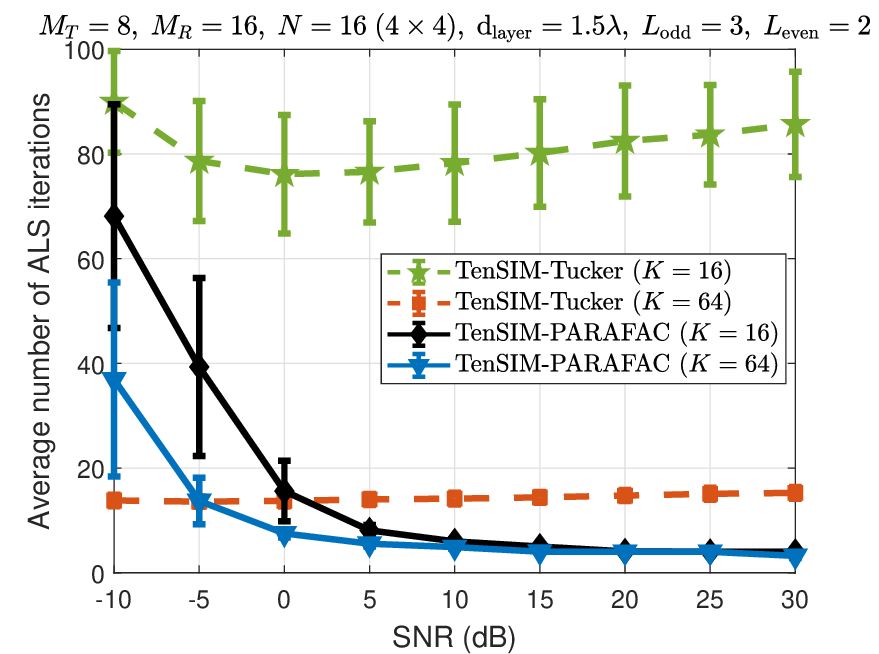}
\caption{Average number of ALS iterations versus SNR for the proposed TenSIM-PARAFAC and TenSIM-Tucker estimators, illustrating their different convergence behavior under the same stopping criterion.}
\label{fig:iters-snr}
\end{figure}
To study the convergence behavior of the proposed TenSIM framework, Fig.~\ref{fig:iters-snr} shows the average number of ALS iterations required for convergence as a function of the SNR for both the TenSIM-PARAFAC and TenSIM-Tucker estimators. The results reveal two markedly different convergence behaviors. The TenSIM-Tucker estimator exhibits an almost SNR-independent convergence pattern, requiring approximately a fixed number of ALS iterations across the entire SNR range. This behavior indicates that its convergence is primarily governed by the conditioning of the Tucker LS subproblems and by the structural coupling introduced by the SIM core matrix, rather than by the additive noise level. In particular, once the Tucker model structure is established, the ALS updates tend to stabilize around a similar convergence trajectory regardless of the noise variance. In contrast, the TenSIM-PARAFAC estimator shows a strong dependence on the SNR. At low SNR values, the number of iterations required for convergence is significantly larger, reflecting the increased difficulty in accurately estimating the factor matrices in the presence of strong noise. As the SNR increases, the number of iterations decreases substantially. These results highlight an important algorithmic distinction between the two estimators. While the Tucker-based approach exhibits a more stable but relatively fixed convergence behavior, the PARAFAC-based estimator benefits more directly from improved pilot signal quality, resulting in faster convergence at moderate-to-high SNR levels. This difference complements the complexity analysis presented earlier and provides additional insight into the convergence characteristics of our TenSIM algorithms.

\begin{figure}[t]
\centering
\includegraphics[width=\columnwidth]{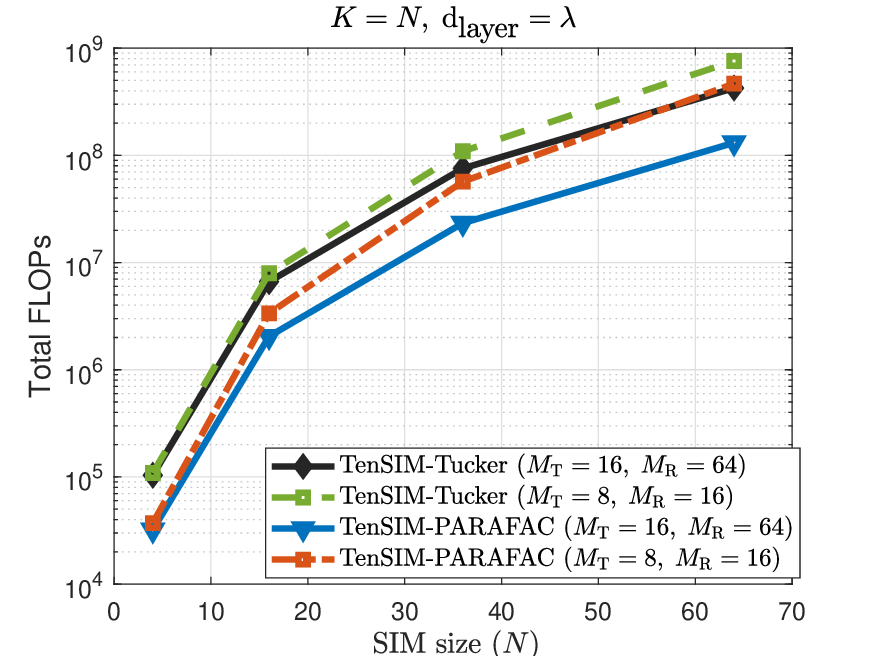}
\caption{Computational complexity of the proposed TenSIM-PARAFAC and TenSIM-Tucker estimators versus the number of SIM elements $N$, reported in floating-point operations for two antenna configurations.}
\label{fig:complexity-flops}
\end{figure}

In Figure \ref{fig:complexity-flops}, we compare the computational complexity of the proposed TenSIM-PARAFAC and TenSIM-Tucker estimators as a function of the SIM size $N$, measured in terms of the total number of floating-point operations (FLOPs). The results are shown for two antenna configurations, $(M_T,M_R)=(8,16)$ and $(M_T,M_R)=(16,64)$, assuming $K=N$ training blocks. As predicted by the complexity analysis in Section~V-E, the computational cost of both estimators grows rapidly with $N$, reflecting the quadratic dependence of the tensor construction and LS updates on the SIM dimension. However, the TenSIM-Tucker estimator consistently exhibits a higher computational cost than the TenSIM-PARAFAC one. This difference originates from the additional $N^4$ term introduced by the SIM core coupling matrix in the Tucker-based formulation, which is absent in the PARAFAC model. As a result, the complexity gap between the two estimators increases with $N$, highlighting the more favorable scaling of the PARAFAC-based formulation for large SIM arrays. The antenna configuration primarily shifts the curves vertically, since the dominant LS operations scale proportionally with $(M_T+M_R)$. Nevertheless, the relative behavior of the two estimators remains essentially unchanged, indicating that the SIM dimension is the main factor governing the computational complexity. These results corroborate the analytical complexity expressions derived in Section~V-E and confirm the improved scalability of the TenSIM-PARAFAC estimator. In particular, the results reflect the different asymptotic scalings of the two estimators: while the dominant computational cost of the TenSIM-PARAFAC estimator grows approximately as $O(N^3)$, the TenSIM-Tucker estimator exhibits a steeper $O(N^4)$ growth due to the additional SIM core coupling operations, indicating that the SIM dimension is the main factor governing the computational complexity.

\begin{figure}[t]
\centering
\includegraphics[width=\columnwidth]{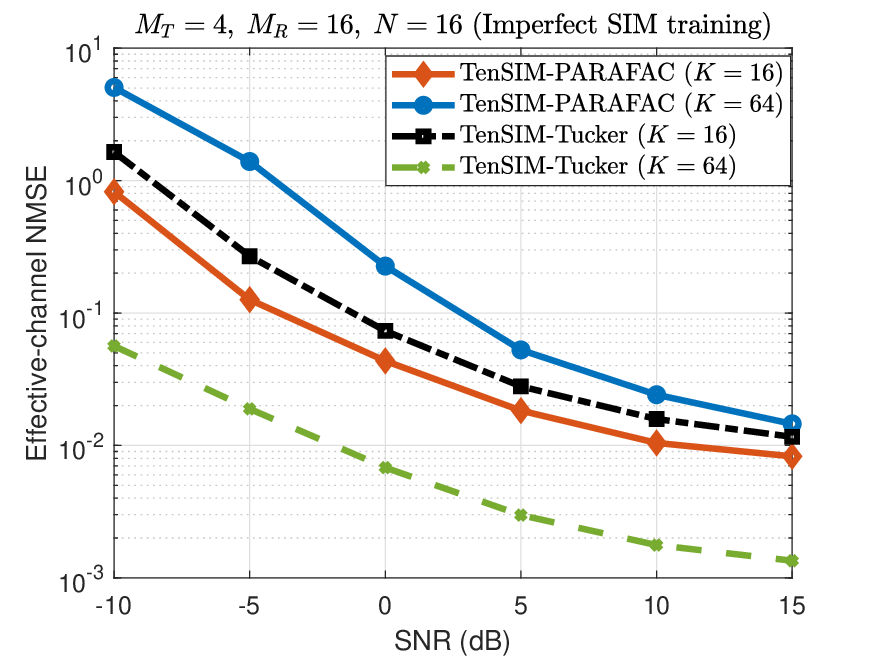}
\caption{Effective-channel NMSE versus SNR in an imperfect training scenario where the unknown SIM training matrix is initialized from its nominal value and then jointly updated with the channel matrices for different pilot lengths $K$.}
\label{fig:tensim_blind}
\end{figure}
Finally, the performance of an imperfect training scenario (the blind-SIM variant discussed in Section~\ref{sec:blind-SIM-variant}) in which the SIM training matrix is not perfectly known at the Rx is investigated in Fig.~\ref{fig:tensim_blind}. This situation may arise in practical deployments due to hardware impairments, such as phase noise, phase quantization errors, calibration mismatch, or other deviations between the programmed and the actually implemented SIM coefficients. To emulate this effect, taking the odd-layer TenSIM-PARAFAC case as a representative example, the SIM training matrix is modeled as follows:
\begin{equation}
\mathbf \Phi = \mathbf \Phi_0 + \sqrt{\sigma}\,\mathbf E,
\end{equation}
where $\mathbf \Phi_0$ denotes the nominal SIM training matrix known at the Rx, $\mathbf E$ is a random perturbation matrix with complex-exponential entries and uniformly drawn phases, and $\sigma=10^{-2}$ controls the impairment level. In this imperfect SIM configuration, TenSIM is initialized with the nominal matrix $\mathbf \Phi_0$ (as discussed in Section~\ref{sec:training_design}), but the actual training factor $\mathbf \Phi$ is then estimated jointly with the channel matrices. This corresponds to activating the optional SIM training-factor updates in Algorithm~\ref{alg:unified_sim_als}, yielding a full three-step ALS routine that alternates between the channel factors and the now-unknown SIM training matrix. The results in Fig.~\ref{fig:tensim_blind} show that the effective-channel NMSE decreases with increasing SNR for both TenSIM methods, confirming that the proposed solutions remain useful under training uncertainty. Increasing the number of pilot blocks from $K=16$ to $64$ improves the performance of both TenSIM-PARAFAC and TenSIM-Tucker due to additional training diversity. The results also showcase that the two TenSIM methods respond differently. More specifically, TenSIM-PARAFAC benefits from its simpler separable structure, whereas TenSIM-Tucker exploits the richer coupling model when sufficient training diversity is available. Overall, Fig.~\ref{fig:tensim_blind} confirms that the proposed TenSIM framework can accommodate imperfect training conditions and highlights the need for additional pilot blocks to compensate for the additional uncertainty introduced by imperfect SIM training. 

\subsection{Discussion} Our results highlight a fundamental trade-off. Increasing the number of SIM layers enhances the electromagnetic wave-manipulation capability of the cascade by enabling multiple successive transformations of the propagating field. However, the same physical mechanisms that increase wave-control flexibility also introduce stronger interlayer coupling, which is reflected mathematically in the propagation matrix $\mathbf{W}_{M+1}$ of the Tucker core tensor. As the SIM aperture grows or the distance between layers increases, this operator tends to exhibit a more uneven singular-value spectrum, leading to less well-conditioned updates during LS iterations. In contrast, the odd-layer architecture effectively absorbs the inter-layer propagation operators into the effective channel factors $\mathbf{Z}_G$ and $\mathbf{Z}_H$, resulting in a PARAFAC model with an identity core tensor. This structural simplification removes the direct influence of the propagation operator from the ALS updates and helps explain the greater numerical robustness observed in the odd-layer estimator across the reported experiments. Overall, these results suggest that the choice between odd- and even-layer SIM configurations should not be guided solely by electromagnetic design considerations but should also account for the associated channel estimation complexity and numerical stability. In particular, for large SIM apertures or larger inter-layer spacings, the odd-layer configuration may offer a more favorable trade-off between wave-control flexibility and estimation reliability.

\subsection{\revblue{Limitations and Practical Considerations}}
\revblue{TenSIM is most effective when the SIM geometry is calibrated, the inter-layer coupling matrices are known or accurately modeled, the channels remain quasi-static during the pilot interval, and the central-layer training factors provide sufficient rank diversity. TenSIM-PARAFAC is generally preferable for large SIM apertures, fast tracking, or larger inter-layer spacings, because its PARAFAC model has an identity core and lower computational cost. However, it estimates effective factors rather than separately identifying all physical subchannels. TenSIM-Tucker can provide more accurate channel reconstruction when the two central layers are strongly coupled, and enough training blocks are available, but it is more sensitive to the conditioning of $\mathbf W_{M+1}$ and its complexity grows more rapidly with $N$. Both branches can degrade under limited training length, nearly rank-deficient factors, imperfect coupling calibration, phase noise, quantized phase shifts, switching latency, or unknown SIM phase perturbations. The blind/imperfect-training variant mitigates part of this uncertainty by updating the training factor jointly with the channel factors, but it requires additional pilot diversity and is more sensitive to initialization.}

\section{Conclusions}
This paper presented TenSIM, a tensor-based framework for channel estimation in SIM-assisted MIMO communication systems under reduced-complexity training protocols. By exploiting the SIM cascade structure, TenSIM unifies two distinct multilinear models depending on the parity of the number of metasurface layers: a PARAFAC model for odd-layer SIM and a Tucker model for even-layer SIM. These formulations expose the latent factors associated with the Tx-side channel, the Rx-side channel, and the training-dependent SIM coefficients in a compact, structured manner. Building on these models, we derived TenSIM ALS estimators tailored to the odd- and even-layer cases and established their identifiability and computational properties. Our numerical investigations corroborated our TenSIM modeling framework and showcased the practical trade-offs between the two SIM architectures. In particular, the NMSE, convergence, and complexity trends demonstrated that both TenSIM branches benefit from increased training diversity and improved SNR, while the TenSIM-PARAFAC branch consistently offered a more favorable balance between estimation accuracy, numerical robustness, and computational cost. By contrast, the TenSIM-Tucker branch remained the appropriate structured model when two central layers need to be trained, but its Tucker core introduced greater sensitivity to inter-layer coupling and a steeper growth in complexity with increasing SIM size. These observations highlight that the choice between odd- and even-layer SIM configurations within TenSIM should be guided by the resulting estimation stability and implementation cost. Overall, TenSIM connects electromagnetic-layer programmability to identifiable communication-layer models, while enabling decoupled estimates of the involved channels from a common tensor-estimation viewpoint. Future work includes data-aided and structure-exploiting channel estimation methods, robust designs under model mismatch, and implementations that leverage sparsity to further reduce training overhead and computational complexity.

\renewcommand\baselinestretch{.95}

\end{document}